\documentclass[12pt,a4paper]{article}

\usepackage[utf8]{inputenc}
\usepackage[T1]{fontenc}
\usepackage{microtype}
\usepackage{fourier}
\usepackage{tensor}

\usepackage{tikz}
\usepackage[margin=10pt,font=small,labelfont=bf,labelsep=period]{caption}
\usepackage{amsmath}
\usepackage{amsthm}
\usepackage{amssymb}
\usepackage{mathscinet}
\usepackage[toc,page]{appendix}
\numberwithin{equation}{section}

\usepackage[nottoc,notlot,notlof]{tocbibind}

\usepackage{aliascnt}
\usepackage[colorlinks,linkcolor=blue,citecolor=blue]{hyperref}

\theoremstyle{plain}

\newtheorem{theorem}{Theorem}[section]

\newaliascnt{lemma}{theorem}
\newtheorem{lemma}[lemma]{Lemma}
\aliascntresetthe{lemma}

\newaliascnt{proposition}{theorem}
\newtheorem{proposition}[proposition]{Proposition}
\aliascntresetthe{proposition}

\newaliascnt{conjecture}{theorem}

\aliascntresetthe{conjecture}

\newaliascnt{corollary}{theorem}
\newtheorem{corollary}[corollary]{Corollary}
\aliascntresetthe{corollary}

\theoremstyle{definition}

\newaliascnt{definition}{theorem}

\aliascntresetthe{definition}

\newaliascnt{example}{theorem}

\aliascntresetthe{example}

\newaliascnt{remark}{theorem}

\newtheorem{remark}[remark]{Remark}
\aliascntresetthe{remark}

\newcommand{\R}{\mathbf{R}}
\newcommand{\C}{\mathbf{C}}

\renewcommand{\epsilon}{\varepsilon}
\newcommand{\g}{\gamma}

\renewcommand{\a}{\alpha}
\renewcommand{\b}{\beta}
\renewcommand{\g}{\gamma}

\renewcommand{\l}{\lambda}

\newcommand{\s}{\sigma}

\newcommand\inv{^{-1}}

\newcommand{\be}{\begin{equation}}
\newcommand{\ee}{\end{equation}}
\newcommand{\bea}{\begin{eqnarray}}
\newcommand{\eea}{\end{eqnarray}}
\newcommand{\bes}{\begin{equation*}}
\newcommand{\ees}{\end{equation*}}
\newcommand{\beas}{\begin{eqnarray*}}
\newcommand{\eeas}{\end{eqnarray*}}
\renewcommand{\=}{\ =\ }
\newcommand{\bsm}{\left[\begin{matrix}}
\newcommand{\esm}{\end{matrix}\right]}
\newcommand{\intv}{\int^1_{-1}}

\newcommand{\wh}{\widehat}

\newcommand{\la}{\langle}
\newcommand{\ra}{\rangle}
\newcommand{\bz}{{\bf 0}}
\newcommand{\pint}{\int_{-1}}
\newcommand{\nin}{\noindent}
\newcommand{\sms}{\smallskip}
\newcommand{\mds}{\medskip}

\newcommand{\D}{\Delta}

\newcommand{\abs}[1]{\left\lvert #1 \right\rvert}
\newcommand{\avg}[1]{\bigl\langle #1 \bigr\rangle}

\DeclareMathOperator{\sech}{sech}
\DeclareMathOperator{\sgn}{sgn}

\renewcommand{\d}{\delta}
\newcommand{\vp}{\varphi}

\renewcommand{\l}{\lambda}

\renewcommand{\=}{\ =\ }

\newcommand{\bo}{{\bf 1}}
\renewcommand{\H}{{\bf H}}

\begin{document}

\title{A 2-component Camassa-Holm equation, Euler--Bernoulli Beam Problem and Non-Commutative Continued Fractions}
\author{Richard Beals
\thanks{Department of Mathematics, Yale University, New Haven, CT 06520, USA; richard.beals@yale.edu}
  \and
  Jacek Szmigielski\thanks{Department of Mathematics and Statistics, University of Saskatchewan, 106 Wiggins Road, Saskatoon, Saskatchewan, S7N\,5E6, Canada; szmigiel@math.usask.ca}}

\date{May 24, 2021}

\maketitle

\begin{abstract}
 A new approach to the Euler-Bernoulli beam based on an inhomogeneous matrix
 string problem is presented.
 Three ramifications of the approach are developed:
 \begin{enumerate}
 \item motivated by an analogy with the Camassa-Holm equation a class of
 isospectral deformations of the beam problem is formulated;
 \item a reformulation of the matrix string problem in terms of a certain compact operator 
is used to obtain basic spectral properties of the inhomogeneous matrix string problem
with Dirichlet boundary conditions;
 \item the inverse problem is solved for the special case of a discrete Euler-Bernoulli beam.
 The solution involves a non-commutative generalization of Stieltjes' continued fractions,
 leading to the inverse formulas expressed in terms of ratios of Hankel-like determinants.
 \end{enumerate}
\end{abstract}

\tableofcontents{}

\section{Introduction}
In 1993, Camassa and Holm
\cite{camassa-holm} discovered the shallow water equation
\begin{equation}
  \label{eq:CH}
  m_t + (um)_x+u_xm = 0
  ,\qquad
  m = u - u_{xx},  
\end{equation}
where subscripts denote partial derivatives.  The most attractive novel property of 
\eqref{eq:CH} is that it supports non-smooth soliton solutions.  
These peaked solitons (\textit{peakons}) are obtained from the ansatz 
\begin{equation}
  \label{eq:CHpeakons}
  u(x,t) = \sum_{j=1}^d m_j(t) \, e^{-\abs{x - x_j(t)}}, 
\end{equation}
involving amplitudes $m_j(t)$ and positions $x_j(t)$ depending smoothly on time.  It was shown 
in \cite{camassa-holm} that $u$ given by \eqref{eq:CHpeakons} is a weak solution to the 
Camassa--Holm equation \eqref{eq:CH} if and only if the positions~$x_j(t)$ and
amplitudes $m_j(t)$ satisfy the Hamiltonian system
\begin{equation}
  \label{eq:CH-peakon-ode}
  \dot x_j = \frac{\partial H}{\partial m_j} = u(x_j)
  ,\qquad
  \dot m_j = -\frac{\partial H}{\partial x_j} = -m_j u_x(x_j), 
  \end{equation} 
 with Hamiltonian 
\begin{equation*}
  H(x_1,\dots,x_d,m_1,\dots,m_d)
  = \frac12 \sum_{i,j=1}^d m_i m_j e^{-\abs{x_i-x_j}},  
 \end{equation*}
and the convention that $u_x(x_i)=\avg{D_x u}(x_i)$ where $\avg{f}(x_i)$ means the arithmetic 
mean of the left and right hand limits of $f$ at $x_i$.   

The general solution of \eqref{eq:CH-peakon-ode} (for arbitrary~$d$)
was constructed in \cite{BSS-acoustic,beals-sattinger-szmigielski:stieltjes,beals-sattinger-szmigielski:moment} using inverse spectral methods. 
The main premise of these papers was the realization that the CH equation can be viewed 
as an isospectral deformation of the classical inhomogeneous string problem studied in the $1950$s 
by M.G. Krein \cite{krein-string, krein-stieltjes, kackrein} and, subsequently, by Dym and 
McKean in \cite{mckean};  for a review see \cite{EKT}.  
 In particular, the peakon solutions 
\eqref{eq:CHpeakons} were shown in \cite{beals-sattinger-szmigielski:stieltjes} to be directly linked to an isospectral deformation of discrete 
strings, i.e. those strings for which the mass density is a linear combination of point masses. 
Krein had observed long ago that, in the case of discrete strings, the  inverse string problem can be 
solved explicitly  by using results of Stieltjes on continued fractions \cite{stieltjes}.  
Stieltjes'  methods were applied in \cite{beals-sattinger-szmigielski:stieltjes} 
to derive the determinantal formulas for the amplitudes and positions of 
peakons and, in turn, to determine the asymptotic behaviour of peakon solutions.  

It has been known at least since the fundamental paper \cite{lax-kdv} by Lax 
on the Korteweg--de Vries equation that the isospectral deformations of boundary value problems 
may lead to  interesting non-linear 
equations.  The boundary value problem for the KdV equation is given by the Schr\"{o}dinger equation 
on the whole real axis, in which case the spectrum is a union of continuous and discrete spectra. 

In the late $1970$s, in a series of interesting papers \cite{Sabatier-Constants, Sabatier-Evolution, 
Sabatier-NEEDS}, Sabatier put forward an idea that spectral problems with just discrete spectrum 
can also be a source of interesting non-linear problems derived \textit{via} isospectral deformations.  
In particular, the inhomogeneous string boundary value problem was singled out  as a potential 
source of interesting boundary problems with the discrete spectrum, to be subjected to isospectral 
deformations.  Sabatier studied a very limited class of spectral deformations, applicable only to 
strictly positive, smooth densities.  The situation changed with the discovery of the CH equation 
\eqref{eq:CH} and subsequent realization in  
\cite{BSS-acoustic,beals-sattinger-szmigielski:stieltjes,beals-sattinger-szmigielski:moment} that the 
CH equation is, in disguise, an isospectral deformation of an inhomogeneous string boundary value 
problem.  This line of research was later generalized to other equations from the CH family 
(\cite{LS-DegasperisCubic, colville-gomez-sz, GS-string}). In the introduction to 
\cite{beals-sattinger-szmigielski:moment}, the authors speculated, 
guided by the work of Barcilon \cite{barcilon-beam-royal} on the Euler-Bernoulli beam problem, 
that higher order boundary value problems might provide an equally rich environment for
isospectral deformations.  This task required a new look at the beam boundary 
value problem that would be naturally amenable to a Lax type deformation.  The present paper aims 
to offer a new approach to the beam boundary value problem by rephrasing it as a matrix inhomogeneous 
string problem.  

In the remainder of this section we outline the main results of this paper and provide a context 
for some of the techniques used in our arguments.

In \autoref{2CH}  we propose a simple derivation of a system of nonlinear equations that generalizes 
the CH equation to a new two component equation which is structurally of the same type as the CH equation.  
We are keenly aware that there exist other two-component generalizations of the CH equation, for example, 
\begin{align*}
&m_t=(um)_x+u_xm +\rho \rho _x,  & m=u-u_{xx}, \\
&\rho_t=(\rho u)_x, &
\end{align*} 
first derived using tri-hamiltonian methods by P.J. Olver  and Rosenau in \cite{olver-rosenau-triH} 
and studied, for example, in \cite{holm-ivanov}.   Other generalizations have been proposed as well 
\cite{xia-qiao}.

The equation we propose (see \eqref{eq:evolBCH})  takes the form
\begin{align*} 
&n_t=(un)_x+u_x n +vn,  &&m_t=(um)_x+u_xm-vm, \\
&u_{xx}-4 u=n+m,  &&v_x=n-m,   
\end{align*} 
and comes from a matrix valued Lax pair, structurally identical to the original CH case, 
but involving two measures $m$ and $n$  rather than one.  

In \autoref{sec:finiteinterval}, using a Liouville transformation, we map the problem to a finite 
interval, following a similar procedure used  in \cite{BSS-acoustic} to study an acoustic scattering 
problem.  We note that the transformed $x$-member of the Lax pair (see \eqref{eq:x-LaxB}) is a matrix 
version of an inhomogeneous string boundary value problem.  

In \autoref{sec:EBbeam} we review the pertinent facts about the Euler--Bernoulli beam problem, which 
we show can be reformulated as a string problem with a matrix density -- the same system already 
encountered in \autoref{sec:finiteinterval}.  We then study 
the basic spectral properties of that matrix string problem.  

In \autoref{sec:Spectral T} we reformulate the Euler--Bernoulli beam problem as a standard spectral 
problem for a compact (in fact trace-class) operator $T$ and we study the properties of that 
operator on an appropriate Hilbert space.  In particular, we establish basic properties of 
the resolvent of that operator.  

In \autoref{sec:Wronskians}, in \autoref{Kl3}, we derive a closed-form expression for the resolvent 
of $T$ and introduce a special element of the resolvent, the Weyl function, which plays a key role in 
the formulation of the inverse problem.  

In \autoref{sec:discrete beam} we analyze the spectral problem for a discrete Euler-Bernoulli beam, 
that is, a beam in which both measures $m$ and $n$ are chosen to be finite sums of point masses..  
This is a beam counterpart of Stietljes' string.  Similar to what occurs for the string problem, the 
Weyl function admits a continued fraction expansion (\autoref{prop:WStieltjes}), albeit with 
non-commuting coefficients.  The general concept of continued fractions over non-commutative rings 
goes  back to Wedderburn \cite{wedderburn}, but our special non-commutative case is a direct 
generalization of continued fractions of Stieltjes' type, this time associated to a discrete inverse 
beam problem, rather than a string problem.  
The continued fraction expansion can be rephrased in terms of non-commutative Pad\'{e} approximations and 
we formulate Pad\'{e} approximation conditions needed for the inverse problem.

In \autoref{sec:inverse problem} we explicitly solve the inverse problem for the discrete 
Euler-Bernoulli beam.  To this end we construct a sequence of non-commutative Pad\'{e} approximations 
using the Weyl function, or, to be more precise, the spectral measure,  as an input data and, 
in the end, we recover  the discrete measures $m$ and $n$.  
In the process of solving the inverse problem we are prompted to introduce several variations on the 
Hankel moment matrix that are germane to the 
inverse beam problem.  The final formulas bear a remarkable resemblance to string formulas (see 
\cite{beals-sattinger-szmigielski:moment}) with the proviso that in the beam problem the usual 
Hankel determinants of the moment matrix are replaced by determinants of  suitably reduced 
Hankel matrices of moments.  

The paper concludes with (Appendix) \autoref{app1}, in which we provide 
a detailed analysis of the Lax pair parametrization giving rise to \eqref{eq:evolBCH}.

\section{A $2$-Component CH equation}\label{2CH} 
The Camassa-Holm (CH) equation \eqref{eq:CH} is the compatibility condition for a pair 
of scalar equations on the line, which we take for simplicity to have the form 
\begin{equation} \label{eq:xt-LaxCH}
\psi_{xx}=(1+\lambda m) \psi, \qquad \psi_t=a\psi +b\psi_x, 
\end{equation} 
where $-\infty<x<\infty$, and the subscripts in $x$ and $t$ represent distribution derivatives 
in $x,t$ respectively.  

\begin{remark} In general, the derivatives we consider are distribution derivatives.
It will sometimes be convenient to use $D_x$, $D_t$, or simply $D$ in the case of one variable.
\end{remark}

The CH flow corresponds to the choice 
$$
b=u+\frac1\lambda, \quad 
a=-\frac {u_x}{2}.
$$ 
In this case the compatibility condition reads 
\begin{equation} \label{eq:CHZCC}
m_t=(um)_x+u_xm, \qquad (u_{xx}-4 u)_x=2m_x. 
\end{equation} 
This compatibility condition holds even if $m$ is a measure, 
in which case $u_x$ is of bounded variation, $u$ is continuous, and 
the term $u_xm$ means that on the singular support of $m$ the multiplier is 
taken to be $u_x(a)=\avg{u_x}(a)$, where $\avg{f}(a)$ is $\frac12[f(a_-)+f(a_+)]$. 

\begin{remark} 
Our choice of the coefficients in \eqref{eq:xt-LaxCH} differs slightly from the original 
Lax pair in \cite{camassa-holm}.  This change results in a  different relation between 
$u$ and $m$ as one can see by comparing 
\eqref{eq:CH} and $\eqref{eq:CHZCC}$. 
\end{remark} 

Now, we consider a two-component version
\begin{subequations} 
\begin{align} 
\Psi_{xx}&=(\mathbf{1}+\lambda M) \Psi, \qquad M=\begin{bmatrix} 0&n\\m&0 \end{bmatrix}, \label{eq:x-LaxBCH} \\
\Psi_t&=a\Psi+b\Psi_x,  \label{eq:t-LaxBCH}
\end{align}
\end{subequations} 
where $\Psi, a$ and $b$ are $2\times 2$ matrix functions of $x,t, \lambda$ and $\bo$
denotes the $2\times 2$ identity matrix.   We assume for now that $a,b$ have only terms 
of degree $0, -1$ in $\lambda$.  As shown in the Appendix, if we add the 
assumption that the matrices $a$ and $b$ are bounded  as $x\to\pm\infty$, then,
up to a normalization, they have the form
\begin{equation} \label{eq:abBCH}
b=\begin{bmatrix} u&0\\
0&u \end{bmatrix}+\frac1\lambda \begin{bmatrix}0&1\\1&0 \end{bmatrix}, 
\qquad \qquad 
a=-\frac12 \begin{bmatrix}u_x-v&0\\0&u_x+v \end{bmatrix}.  
\end{equation} 
The compatibility conditions split into two constraints 
\begin{equation*} 
(u_{xx}-4 u)_x=(n+m)_x, \qquad v_x=n-m, 
\end{equation*} 
and a system of evolution equations
\begin{equation} \label{eq:evolBCH} 
n_t=(un)_x+u_x n +vn, \qquad m_t=(um)_x+u_xm-vm.  
\end{equation} 
Furthermore, if we are interested in bounded $u$ and compactly supported $m,n$, then we 
can replace one of the constraints with a more restrictive 
one, keeping the other constraint intact,  
\begin{equation} \label{eq:constraintsBCH}
u_{xx}-4 u=n+m, \qquad v_x=n-m.  
\end{equation} 
We assume that $M$ is compactly supported and take as solutions of \eqref{eq:constraintsBCH} 
the particular choices 
\begin{align}
u(x,t)&=-\frac14 \int_{-\infty}^\infty e^{-2\abs{x-y}}[m(y,t)+n(y,t)] \, dy, \\
v(x,t)&=\frac12 \int_{-\infty}^\infty \sgn (x-y) [n(y,t)-m(y,t)]\, dy. 
\end{align} 
It follows that 

\begin{equation} \label{eq:uvass}
\begin{matrix}\lim_{x\rightarrow \pm \infty} u(x,t)=0=\lim_{x\rightarrow \pm \infty}u_x(x,t)\\ 
\lim_{x\rightarrow \pm \infty}v(x,t)=\pm \frac12  \int_{-\infty}^\infty[n(y,t)-m(y,t)]\, dy. 
\end{matrix}\end{equation}
 
\begin{remark} The integrals appearing in this paper are all Stieltjes integrals, but we find it 
more convenient to write them as  $\int f(x)m(x,t)dx$, or, later in the paper, as $\int f(x)dm(x,t)$ .  The former 
notation is analogous to writing the point mass at the origin as $\d(x)dx$.\end{remark}

\begin{proposition} \label{prop:totalmass}
The integral 
\begin{equation*} 
\int_{-\infty}^\infty [m(y,t)+n(y,t)]\, dy
\end{equation*} 
is independent of $t$.  
\end{proposition} 
\begin{proof} 
According to \eqref{eq:evolBCH}, \eqref{eq:uvass}, 
\begin{equation*} 
\begin{split} 
 D_t \int_{-\infty}^\infty [m(y,t)+n(y,t)]\, dy&=\int_{-\infty}^\infty\{[u(n+m)]_y
+u_y(n+m)+v(n-m)\} \, dy\\
 &=\int_{-\infty}^\infty [u_y u_{yy}-4u_y u+vv_y]\, dy\\
 &=\frac12 \int_{-\infty}^\infty [u_y^2-4u^2+v^2]_y \, dy=0.
 \end{split}
 \end{equation*} 
 
 \end{proof} 
\begin{remark} Note, however, that separately $\int_{-\infty}^\infty m dy$ and 
$\int_{-\infty}^\infty n dy$ are not conserved.  Indeed, it follows from 
\eqref{eq:constraintsBCH} that
\begin{equation*} 
D_t \int_{-\infty}^\infty [n-m] \, dy=-4 \int_{-\infty}^\infty uv\, dy .  
\end{equation*} 
\end{remark} 

\section{Transfer to an interval} \label{sec:finiteinterval}
As in  \cite{beals-sattinger-szmigielski:stieltjes}, we transfer the problem on the  real 
axis to the interval $[-1,1]$.  In this section we will 
refer to the variable on the real axis as $\xi$ and its counterpart on $[-1,1]$ by $x$.
The functions originally defined on the real axis will carry a tilde.  Thus, for example,  
$M$ from the previous section will be denoted by $\tilde M$.  
We set 
\begin{equation*} 
x=\tanh \xi, \qquad D_x=\cosh^2 \xi D_\xi
\end{equation*} 
Define, for general $f: \R\rightarrow \R$, 
\begin{equation*} 
f^*(\tanh \xi)=f(\xi). 
\end{equation*} 
Note that $\sech^2 \xi=1-x^2$.  Then a straightforward computation shows that 
\begin{equation*} 
(1-x^2)^{\frac32} D_x^2 (f^* (1-x^2)^{-\frac12})=[(D_\xi^2-1)f]^*. 
\end{equation*} 
Therefore
\begin{equation*} 
[(D_\xi^2-\mathbf{1}-\lambda \tilde M)f]^*=(1-x^2)^\frac32[D_x^2 -\lambda (1-x^2)^{-2} 
\tilde M ](f^*(1-x^2)^\frac12) 
\end{equation*}
and solutions to \eqref{eq:x-LaxBCH}, after flipping $x$ with $\xi$ and $M$ with $\tilde M$, 
correspond to solutions of 
\begin{equation} \label{eq:x-LaxB} 
\Phi_{xx} =\lambda M \Phi
\end{equation} 
under the map 
\begin{equation*} 
\Phi=(1-x^2)^\frac12 \Psi^*, \qquad M=(1-x^2)^{-2} \tilde M ^*. 
\end{equation*} 
The flow equation \eqref{eq:t-LaxBCH} implies 
\begin{equation*} 
\begin{split} 
D_t\Phi=[(\cosh \xi)^{-1} D_t \Psi]^*&=[(\cosh \xi)^{-1} (\tilde a \Psi+\tilde b D_\xi \Psi]^*\\
&=[(\cosh \xi)^{-1} \tilde a \Psi+(\cosh \xi)^{-3} \tilde b \cosh ^2 \xi D_\xi \Psi]^*\\
&=\tilde a^* \Phi +(1-x^2)^\frac32 \tilde b^*D_x\big((1-x^2)^{-\frac12} \Phi\big)\\
&=\tilde a^* \Phi +(1-x^2)\tilde b^* \big[ D_x+\frac{x}{1-x^2}\big] \Phi\\
&=[\tilde a ^*+x\tilde b^*] \Phi +(1-x^2) \tilde b^* D_x \Phi. 
\end{split} 
\end{equation*} 
Thus 
\begin{equation}\label{eq:t-LaxB}
\Phi_t=a\Phi+ b\Phi_x, 
\end{equation} 
where 
\begin{equation*}
a=\tilde a^*+x\tilde b^*, \qquad b=(1-x^2)\tilde b^*, 
\end{equation*} 
or, more explicitly, 
\begin{equation} \label{eq:abuv}
\begin{split}
a&=\frac12 \begin{bmatrix} -u_x+v&0\\0&-u_x-v \end{bmatrix} -\frac{1}{2\lambda}
\begin{bmatrix} 0&\beta_x\\\beta_x&0 \end{bmatrix},  \\
b&=\begin{bmatrix} u&0\\0&u \end{bmatrix} +\frac1\lambda \begin{bmatrix} 0& \beta\\\beta&0 
\end{bmatrix},  \\
u&=(1-x^2)\tilde u^*, \qquad v=\tilde v^*, \qquad \beta=1-x^2. 
\end{split} 
\end{equation} 

It can be shown that 
\begin{equation} \label{eq:uv-Bconstraints}
u_{xxx}=(\beta (m+n))_x+\beta_x (m+n), \qquad v_x=\beta(n-m), 
\end{equation} 
and that the flow of $M$ takes the same form as the flow of $\tilde M$, namely, 
\begin{equation} \label{eq:mn-Bevol}
n_t=(un)_x+u_xn +vn, \qquad m_t=(um)_x+u_x m-vm. 
\end{equation} 
In fact \eqref{eq:uv-Bconstraints} and \eqref{eq:mn-Bevol} are consequences of the 
compatibility conditions for the Lax pair 
\eqref{eq:x-LaxB} and \eqref{eq:t-LaxB}.  It follows from 
\eqref{eq:uvass} and \eqref{eq:abuv} that 
\begin{equation} \label{eq:uvassB} 
u(\pm 1)=u_x(\pm 1)=0, \qquad v(-1)=-v(+1). 
\end{equation} 
\begin{proposition} \label{prop:totalmassB}
The integral 
\begin{equation*} 
\int_{-1}^1\beta [m+n]\, dx
\end{equation*} 
is independent of $t$.  
\end{proposition} 
\begin{proof} 
According to \eqref{eq:mn-Bevol}, \eqref{eq:uv-Bconstraints} and \eqref{eq:uvassB}, 
\beas
&&D_t \int_{-1}^1\beta(x) [m(x,t)+n(x,t)]\, dx\\
&&\qquad\=\int_{-1}^1 \beta(x)\{[u(n+m)]_x+u_x(n+m)+v(n-m)\} \, dx\\
&&\qquad\=\int_{-1}^1[-\beta_x u(m+n)-u(\beta(m+n))_x+vv_x]\, dx\\
&&\qquad\=\int_{-1}^1 (-u_{xxx} u +v_x v) \, dx\\
&&\qquad\=\frac12 \int_{-1}^1 [u_x^2+v^2]_x \, dx=\frac12(u_x^2+v^2 )\big| _{-1}^1=0.
\eeas
\end{proof} 

We end this section by stating integral formulas for $u$ and $v$, which, as can be 
easily checked, provide the unique solutions of \eqref{eq:uv-Bconstraints} subject to 
boundary conditions 
\eqref{eq:uvassB}. 
\begin{proposition} \label{prop:if-uv}
\begin{align*} 
u(x, t)&=-\int_{-1}^1 G_D(x,y)^2 (m(y,t)+n(y, t))\, dy,\\
 v(x,t)&=\int_{-1}^1 \sgn(x-y) G_D(y,y)  (n(y,t)-m(y, t)) \, dy. 
\end{align*} 
\end{proposition} 

Here
\begin{equation} \label{eq:Green}
G_D(x,y)=\frac12 \begin{cases} (1+x)(1-y), \quad  x<y\\ (1-x)(1+y), \quad y<x \end{cases}
\end{equation}  
is the Green's function of the classical Dirichlet string problem $$-D_x^2 f=\lambda \rho f, 
\quad ~f(-1)=f(1)=0. $$

\section{The Dirichlet problem for a beam}\label{sec:EBbeam}

Vibrations of a beam parametrized by the interval $-1\le x\le1$ are characterized
by the equation 
\be\label{beam}
D^2[r D^2\phi]\=\l^2m\phi,\qquad D=D_x;
\ee
(see \cite{barcilon-beam-royal}  or, for a comprehensive view of vibration problems in engineering, 
\cite{Gladwell}).
What we refer to as a beam problem is often referred to as an {\it Euler--Bernoulli
beam problem\/}.  The two functions (or positive measures) $r,m$ are the flexural
rigidity and mass density.  The spectral parameter $\l^2$ denotes the square of the
frequency.

Setting
\be
D^2\vp_1\=\l n\vp_2,
\ee
where $n=1/r$, \eqref{beam} becomes
\be\label{beam2}
D^2\vp_1\=\l n\vp_2,\qquad D^2\vp_2\=\l m\vp_1.
\ee
The matrix form is 
\be\label{beam-system}
D^2\vp
\=\l M\vp,\qquad \vp\=\begin{bmatrix}\vp_1\\ \vp_2    
\end{bmatrix},\quad M\=\begin{bmatrix} 0&n\\ m&0   \end{bmatrix}.
\ee

We require that $\vp$ be continuous. We want to allow  $m$ and $n$ to be finite positive measures 
on the interval $-1\le x\le 1$.   We assume
\be\label{support}
\begin{matrix}\hbox{(a)\quad The endpoints $x=\pm 1$ are not atoms for $m$ or $n$.}\\
              \hbox{(b)\quad The measures $m$ and $n$ have the same support.}
\end{matrix}
\ee
Then $M\vp$ should be interpreted as
\bes
M\vp\=\bsm \vp_2 n\\ \vp_1 m\esm.
\ees

\begin{remark}Throughout this section and the next 
it is convenient to adopt the (more correct) way
of writing integrals with respect to $m$ and $n$, using $dm(x)$, $dn(x)$ rather than
the symbolic $m(x)dx$, $n(x)dx$.\end{remark}


A partial fundamental solution
$\Phi(x,\l)$ with the property ~~$\Phi(-1,\lambda)=0,\, $ $ D\Phi(-1,\l)=~~\bo$,
the identity matrix, can be constructed in the form
\be\label{Phi}
\Phi(x,\l)\=\sum_{k=0}^\infty\l^k\Phi_k(x)
\ee
with
\bes
\Phi_0(x)\=(1+x)\bo,\qquad \Phi_{k+1}(x)\=\pint^x\left[\pint^yM(z)\Phi_k(z)\,dz\right]\,dy.
\ees
To be specific, we take the integrals to run on the intervals $[-1,x)$ and $[-1,y)$.

\begin{proposition}\label{Phik} The functions $\Phi_k(x)$ are diagonal for even $k$, off-diagonal for
odd $k$.  The non-zero entries $\Phi_{k1}$, $\Phi_{k2}$ are non-negative, positive at $x=1$, and 
satisfy the estimate
\be\label{Phi-est}
0\ \le\  \Phi_{kj}(x)\ \le\ 2\frac{(1+x)^k(\overline m+\overline n)^k}{k^k\,k\,!}.
\qquad \ j=1,2.
\ee
where $\overline m$ and $\overline n$ denote the total mass of $m$ and $n$.
\end{proposition}

\proof The first assertions follow by induction from the construction \eqref{Phi} and the
assumption that $m$ and $n$ are positive measures.  To prove the estimate, let
$p=m+n$, and define $\psi_k$ by
\bes
\psi_0(x)\=1+x,\qquad \psi_{k+1}(x)\=\pint^x\pint^y \psi_k(z)\,dp(z)\,dy.
\ees
Each $\Phi_{kj}(x)$ is bounded by $\psi_k(x)$.  Changing the order of integration and integrating 
first with respect to $y$, the result can be written as
\beas
&& \psi_k(x)\=\pint^x(x-y_k)\psi_{k-1}dp(y_k)\\
&&\=\pint^x\pint^{y_k}\left[(x-y_k)(y_k-y_{k-1})\right]\psi_{k-2}(y_{k-1})\,dp(y_{k-1})\,dp(y_k)\\
&&\=\pint^x\cdots\pint^{y_2}\left[(x-y_k)(y_k-y_{k-1})\cdots(y_2-y_1)\right]
\psi_0(y_1)\,dp(y_1)\,dp(y_2)\cdots\,dp(y_k). 
\eeas
The product in brackets is maximized when the $k$ factors are all the same.  The
interval $[y_1,x]$ has length at most $1+x$, and $\psi_0(x)\le2$, so
\be\label{first-estimate}
0\ \le\ \psi_k(x)\le 2\left[\frac{1+x}{k}\right]^k
\int_{0\le y_1<\dots<y_k< x}dp^{(k)}(y_1,y_2,\dots y_k),
\ee
where $p^{(k)}$ denotes the product measure on the $k$-cube $\{(x_1,\dots x_k): |x_j|\le 1\}$
in $\R^k$.  The domain of integration is one of $k\,!$ pairwise disjoint domains in the 
cube that are obtained by permuting the indices.  Each domain has the same measure, so
\be\label{second-estimate}
0\ \le\ \int_{0\le y_1<\dots<y_k< x}dp^{(k)}(y_1,y_2,\dots y_k)\ \le\ \frac{\left[\pint^xdp(y)\right]^k}
{k\,!}.
\ee
Since  $\int^1_{-1}dp(x)\=\overline m+\overline n$, the estimates \eqref{first-estimate} and
\eqref{second-estimate} imply \eqref{Phi-est}.\qquad$\qed$.

\begin{corollary}\label{Phi-order} 
(a) The function $\Phi(x,\l)$ is continuous in both variables, and entire as a function of
$\l$, for fixed $x$.

\sms\nin (b)
Each entry of $\Phi(1,\l)$ is dominated by 
\be\label{exp-est}
2\sum_{k=0}^\infty \frac{a^{2k}|\l|^k}{(2k)\,!}\ \le\ 2\exp(a|\l|^{1/2}),\qquad a=\sqrt{2(\overline m
+\overline n)}.
\ee
\end{corollary}.

\proof  (a) The estimates \eqref{Phi-est} imply that the series \eqref{Phi} converges uniformly 
on bounded sets in $[-1,1]\times\C$.  

\sms\nin(b)
Since $(2k)\,!\le (2k)^k k\,!$, the estimates \eqref{Phi-est} imply the
bound  \eqref{exp-est}.\qquad$\qed$

\mds
We are interested in the Dirichlet problem for solutions of \eqref{beam-system}:
\be\label{dirichlet}
\vp_1(\pm 1,\l)\=0\=\vp_2(\pm 1,\l).
\ee
A value $\l\in\C$ for which a non-zero solution of \eqref{beam-system},
\eqref{dirichlet} exists will be referred to as a Dirichlet eigenvalue.  Note that
zero is not an eigenvalue.  

\begin{proposition} The Dirichlet eigenvalues $\{\l_\nu\}$ are precisely the zeros 
of $\D$, where
\bes
\D(\l)\=\det \Phi(1,\l).
\ees
They satisfy 
\be\label{zero-sum}
\sum_\nu \frac1{|\l_\nu|}\ <\ \infty.
\ee\end{proposition}

\proof Any solution $\vp$ of the Dirichlet problem with eigenvalue $\l$ is a linear combination 
of the two columns of $\Phi(\cdot,\l)$:
\bes
\vp(x)\=\Phi(x,\l)\,v
\ees 
where $v$ is a constant 2-vector.  The condition at $x=1$ implies that $\Phi(1,\l)v=0$.
Thus the necessary and sufficient condition for the existence of a non-zero
solution (``eigenfunction'') of the Dirichlet problem 
\eqref{beam-system}, \eqref{dirichlet} is that $\det\Phi(1,\l)=0$.

\sms
Corollary \ref{Phi-order} implies that $|\D(\l)|$ is dominated by $\exp(4a|\l|^{1/2})$.
Therefore the zeros $\l_\nu$, if numbered with $|\l_\nu|$ non-decreasing,  satisfy
\be\label{zero-sum0}
|\l_\nu|\ \ge\ c\nu^2
\ee
for some constant $c>0$. \cite[\S 8.21]{titchmarsh}. In particular, 
\eqref{zero-sum} is true.\qquad$\qed$

\mds At the end of this section we will show that the zeros of $\D$ are simple.

\sms
If $\vp$ is a solution of the Dirichlet problem \eqref{beam-system}, \eqref{dirichlet},
then $f=D\vp $ is a solution of 
\be\label{f-problem}
Df\=\l M \vp.
\ee
It follows that $f$ is a function of bounded variation which is continuous at any
point that is not an atom.  In particular, $f$  is continuous at the endpoints $\pm 1$.
This fact, for $f$ and similar functions, justifies the various integration-by-parts
formulas in this and later sections.

\sms
The following identity is fundamental to our discussion of the Dirichlet problem.

\begin{lemma}\label{basic-ID} If $\vp$ is a solution of the Dirichlet problem with
eigenvalue $\l$, and $f=D\vp$, then 
\be\label{basic-id}
\intv [f_1\overline f_2+f_2\overline f_1]\,dx\=-\l\intv [|\vp_1|^2\,dm+|\vp_2|^2\,dn]\ \ne \ 0.
\ee\end{lemma}

\proof It is convenient to write the left side in the form $(D\vp,\s D\vp)$, where
the inner product $(\ ,\ )$ is the $L^2$ inner product for vector-valued functions:
\be\label{L2ip}
(f,g)\=\intv f(x)\cdot \overline{g(x)}\,dx\=
\intv [f_1(x)\overline{g_1(x)}+f_2(x)\overline{g_2(x)}]\,dx,
\ee
and $\s$ is the matrix
\be\label{sigma}
\s\=\begin{bmatrix} 0&1\\ 1&0\end{bmatrix}.
\ee
Then integration by parts gives
\be\label{basic-id2}
(D\vp,\s D\vp)\=-(D^2\vp,\s\vp)\=-\l(M\vp,\s\vp),
\ee
which is \eqref{basic-id}.

If the right side of \eqref{basic-id} is zero, then $\vp_1=0$ on the
support of $m$ and $\vp_2=0$ on the support of $n$.  As a consequence 
\bes
(D\vp,D\vp)\=-\l(M\vp,\vp)\=-\l\intv [\vp_1\overline\vp_2\,dm+\vp_2\overline\vp_1\,dn]\=0,
\ees
implying that $\vp$ is constant, hence zero, a contradiction.\qquad\qed

\begin{theorem}\label{phi-abcd} The Dirichlet eigenvalues satisfy the
following conditions:

\sms\nin (a) If $\l$ is an eigenvalue, so is $-\l$.  

\sms\nin (b) Each  eigenvalue $\l$ is real and its eigenspace has
dimension one.  \end{theorem}

\proof If $\vp=[\vp_1,\vp_2]^t$ is an eigenvector with eigenvalue $\l$, it
follows immediately from \eqref{beam2} that $[\vp_1,-\vp_2]^t$ has eigenvalue
$-\l$.  The complex conjugate $\overline \vp$ has eigenvalue $\overline \l$.  Integration
by parts shows that
\bes
(D^2\vp,\vp)\=(\vp, D^2\vp)\=\overline{(D^2\vp,\vp)},
\ees
so Lemma \eqref{basic-ID} and \eqref{basic-id2} show that 
\bes
0\ \ne\ \l\intv [|\vp_1|^2\,dm+|\vp_2|^2\,dn]\=\overline\l\intv [|\vp_1|^2\,dm+|\vp_2|^2\,dn].
\ees
Therefore $\l=\overline\l$. Then the singular matrix $\Phi(1,\l)$ has positive diagonal entries, 
so it has two eigenvalues, namely  $0$ and ${\rm tr}\,\Phi(1,\l)>0$.  
Thus the eigenspace for $\l$ has dimension 1.  \qquad $\qed$

\begin{proposition}\label{simple-zeros} The zeros of $\D=\det\Phi(1,\l)$ are simple.
\end{proposition}

\proof Write
\bes
\Phi(1,\l)\=\begin{bmatrix}a(\l)&b(\l)\\ c(\l)&d(\l)\end{bmatrix}
\qquad v(\l)\= \begin{bmatrix}d(\l)\\-c(\l)\end{bmatrix},
\ees
and let $\vp(x,\l)=\Phi(x,\l)v(\l)$.  The entries of $\Phi(1,\l)$ are non-zero,
so $\vp\ne0$.  We have $\vp(-1,\l)=\bz$ and 
\bes
\vp(1,\l)\=\begin{bmatrix}\vp_1(1,\l)\\ \vp_2(1,\l)\end{bmatrix}
\=\begin{bmatrix} \D(\l)\\ 0\end{bmatrix}.
\ees
We want to prove that the derivative $D_\l\vp(1,\l)$ does not vanish if $\l$
is an eigenvalue.  Differentiating with respect to $\l$ shows that at an
eigenvalue
\bes
(D_x^2-\l M)D_\l\vp\=M\vp.
\ees
Therefore, by Lemma \eqref{basic-ID}, and because all terms here are real, 
\beas
&&\intv[|\vp_1|^2\,dm+|\vp_2|^2\,dn]\=(M\vp,\s\vp)\\
&&\quad\=(D_x^2D_\l\vp,\s\vp)-(\l M D_\l\vp,\s\vp)\\
&&\quad\=-(D_xD_\l\vp,\s D_x\vp)-\l(M\vp,\s D_\l\vp)\\
&&\quad\=-D_\l \vp\cdot\s D_x\vp\big|^1_{-1} + [(D_\l\vp,\s D_x^2\vp)-(\l M\vp,\s D_\l\vp)]\\
&&\quad\=- D_\l\vp_1(1,\l)D_x\vp_2(1,\l).
\eeas
Therefore $D_\l \D(\l)=D_\l\vp_1(1,\l)\ne 0$.\qquad\qed

\section{Spectral theory} \label{sec:Spectral T}

In this section we rephrase the Dirichlet problem for the Euler-Bernoulli beam as a standard
eigenvalue problem for a compact operator, whose eigenfunctions are the derivatives $D\vp$.

Let $\H$ be the $L^2$ space of vector-valued
functions on the interval $[-1,1]$, with the inner product \ref{L2ip} and norm
\bes
||f||\=(f,f)^{1/2}.
\ees
We also introduce an indefinite form
\be\label{indefinite}
\la f,g\ra\=\intv [f_1(x)\overline{g_2(x)}+f_2(x)\overline{g_1(x)}]\,dx
\=(\s f,g)\=(f,\s g)
\ee
where $\s$ is the matrix \eqref{sigma}.

We know that the Dirichlet eigenvalues $\{\l_\nu\}$ are real, and we take their 
eigenvectors $\{\vp_\nu\}$ to be real as well. Lemma \ref{basic-ID} can be rephrased as
\be\label{f-indefinite}
\la f_\nu,f_\nu\ra\=-\l_\nu\intv[\vp_{\nu,1}^2\,dm+\vp_{\nu,2}^2\,dn]\ne 0.
\ee
Thus $\la f_\nu,f_\nu\ra$ and $\l_\nu$ have opposite signs.  Because of this we index the 
eigenvalues with
\be\label{indexing}
\dots\ <\ \l_{2}\ <\ \l_{1}\ <\ 0\ < \l_{-1}\ <\ \l_{-2}\ <\ \dots\ .
\ee
The calculation that led to \eqref{f-indefinite} leads to two formulations for
$\la f_\nu,f_\mu\ra$ that show
\be\label{fg-zero}
\la f_\nu,f_\mu\ra\=0\qquad \hbox{if\ \  $\mu\ne\nu$.}
\ee

Let $\H_0\subset \H$ consist of the constant functions, and let $\H_1$ be the orthogonal
complement:
\bes
\H_1\=\left\{f\in \H\,:\, \intv f=0\right\}.
\ees
Note that each $f_\nu$ belongs to $\H_1$, since $\intv f_\nu=\vp_\nu(1)-\vp_\nu(0)=0$.

This discussion leads us to introduce two inverses of $D=D_x$ that map $\H$ to $\H$:
\be\label{D0inv}
D_0\inv g(x)\=\begin{cases}\frac12 \pint^x g(y)\,dy-\frac12\int_x^1g(y)\,dy,& g\in \H_1,\\
0,& g\in \H_0,\end{cases}
\ee
and 
\be\label{D1inv}
D_1\inv g(x)\=\frac12\int_{-1}^x(y+1)g(y)\,dy+\frac12\int_x^1(y-1)g(y)\,dy.
\ee

\begin{lemma}\label{Dinverses} For any $g\in\H$,
\bes
D_0\inv g(-1)\=D_0\inv g(1)\=0; \qquad D_1\inv g\in\H_1.
\ees\end{lemma}

\proof The first pair of identities is clear for $g\in\H_1$ and true by definition
for $g\in \H_0$.  The third identity follows from an easy calculation.
\qquad$\qed$

\mds
In view of these remarks, the Dirichlet eigenvalue problem can be reformulated as
\bes
Df\=\l M D_0\inv f,\qquad f\in \H_1.
\ees
or, equivalently,
\be\label{Tproblem}
f\in \H_1,\qquad f\=\l Tf, \qquad T\=D_1\inv M D_0\inv.
\ee
Thus the problem \eqref{beam-system}, \eqref{dirichlet} is 
equivalent to a standard eigenvalue problem for the operator $T$ mapping
$\H_1$ to $\H_1$ or $\H$ to $\H$.
Note that $T=0$ on $\H_0$, since this is true, by definition of $D_0\inv$. 

\begin{lemma} The operator $T$ is a compact operator in $\H$. If $f$ is in
the image of $T$, then $f$ has bounded variation and is continuous at the
endpoints $x=\pm1$.
\end{lemma}

\proof Any inverse  $D\inv$ takes $L^2$ functions to functions that satisfy
a H\"older condition: if $u=D\inv f$ and $x<y$, then $Du=f$ so 
\bes
|u(y)-u(x)|\ \le\ \int_x^y |f(t)|\,dy\ \le\ \left[\int_x^y|f|^2\right]^{1/2}\,
\left[\int_x^y dt\right]^{1/2}\ \le\ ||f||\,(y-x)^{1/2}.
\ees
It follows that the image under $D_0\inv$ of a bounded sequence in $\H$ is uniformly
equicontinuous. The same argument applied to the two summands in \eqref{D0inv}
shows that $|D_0\inv f(x)|\le 2||f||$, so the image is also uniformly bounded.
By the theorem of Ascoli--Arzel\'a there is a uniformly 
convergent subsequence, $\{g_n\}$.  Thus $D_0\inv$ is compact from $\H$ to the space of bounded
continuous functions. The map  $g\to D_1\inv gM$ from this space to $\H$ is bounded, so $T$ is compact.  
Moreover $g\to D_1\inv gM$ maps to
functions of bounded variation; continuity of $Tf$ at the endpoints follows from assumption
(a) of \eqref{support}.
\qquad$\qed$

\begin{lemma}\label{eigen} If $K\ne (0)$ is a closed subspace of $\, \H$ that 
is invariant under $T$, then $K$ contains an eigenfunction of $T$.  \end{lemma}

\proof 
Since $T$ is compact, every non-zero point of the spectrum is an eigenvalue.
This applies also to the restriction of $T$ to the invariant subspace $K$.
Compactness implies that the operator $I-\l T$ is Fredholm with index zero. 
If the restriction to $K$ has no null space, for all $\l\in \C$, then it is invertible,
and the  inverse is entire and bounded at $\l=\infty$, leading to a contradiction.  Therefore
$T$ has an eigenvector in $K$.  \qquad$\qed$

\begin{lemma} The operator $T$ is symmetric with respect to the indefinite
form $\la\ ,\ \ra$.\end{lemma}

\proof We may assume that $f$ and $g$ are in $H_1$ and are smooth.  Since 
$g=D[D_0\inv g]$ and $D_0\inv g(\pm1)=0$ we may integrate by parts to get
\bea
\la Tf,g\ra&=&(Tf,\s DD_0\inv g\ra\=-(DTf,\s D_0\inv g)\nonumber\\
&=& -(M D_0\inv f,\s D_0\inv g)\=-(D_0\inv f, M^t \s D_0\inv g).
\eea
Since $M^t \s$ is diagonal, the last expression is symmetric in $f$ and $g$.
\qquad$\qed$

\mds
Let $\H_T$ be the closure in $\H$ of the range of $T$.  Any solution of
\eqref{Tproblem} will belong to $\H_T$. 

\begin{proposition}\label{completeness} The span of the eigenfunctions $\{f_\nu\}$
is dense in $\H_T$.\end{proposition}

\proof Let $N_T\in \H_T$ be the orthogonal complement of $\{f_\nu\}$ in $\H_T$
with respect to the standard inner product.  Then $\s N_T$ is orthogonal
to $\H_T$ with respect to the indefinite form.  Since $T$ is symmetric,
$\s N_T$  is invariant for $T$. By construction,  $\s N_T$ is orthogonal to every
eigenfunction.  It follows from Lemma \ref{eigen} that $\s N_T=(0)$.
Therefore $N_T=(0)$. \qquad$\qed$

\begin{proposition} Let $N$ be the null space of $T$, i.e.\ $N=\{f\in \H;Tf=0\}$.
Then $N=\s N$. Moreover, $N$ coincides with each of
the subspaces

\sms\nin (a) The orthogonal complement of\ \ $\H_T$ with respect to $(\ ,\ )$;
 
\sms\nin (b) The orthogonal complement of\ \ $\H_T$ with respect to $\la\ ,\ \ra$.

\end{proposition}

\proof Note that $Tf=0$ if and only if either $f\in H_0$, in which case $\s H_0=H_0$,  
or $f\in H_1$ and $M\vp=0$, where $\vp =D_0 \inv f$. 
This, in turn, is equivalent to the conditions that $\vp_1$ vanish on the 
support of $m$  and $\vp_2$ vanish on the support of $n$.  In view of
condition (b) of \eqref{support}, this is equivalent to $M\s\vp=0$,
so $T\s f=0$ whenever $f\in N$.  

The relations between $N$ and the spaces (a), (b) follow from the identities,
valid for every $f,g$ in\ $\H$:
\beas
(f,Tg)&=&\la \s f,Tg\ra\=\la T(\s f), g \ra;\\
\la Tf,g\ra&=&\la f,Tg\ra.
\eeas
The first identity shows that $f$ is in space (a) if and only if $T(\s f)=0$,
which is equivalent to $Tf=0$.  The second shows that $f$ is in space (b) if and
only if $Tf=0$.\qquad $\qed$

\mds
There is a natural decomposition of $\H_T$ into two subspaces that are orthogonal
with respect to the indefinite form:
\bes
\H_T^\pm\=\hbox{closure of the span of $\{f_\nu\,:\,\pm \nu>0\}$.}
\ees

\nin
We introduce a new inner product in $\H_T$ by setting
\be\label{Tip}
(f,g)_T\=\begin{cases}\ \la f,g\ra,\quad f,g\in \H_T^+;\\
-\la f,g\ra,\quad f,g\in \H_T^-;\\
\ 0,\quad\qquad f\in \H_T^+,\ g\in \H_T^-.
\end{cases} 
\ee
Let $\wh \H_T$ be the completion of $\H_T$ with respect to the
norm $||f||_T^2\=(f,f)_T$.  Note that in all cases
\be\label{norms}
(f,f)_T\=|\la f,f\ra|\=|(f,\s f)|\ \le ||f||\,||\s f||\=||f||^2.
\ee
The $\{f_\nu\}$ are clearly an orthonormal basis for this space.
Since $Tf_\nu=\l_\nu\inv f_\nu$, the restriction of $T$ to $\H_T$ extends to a compact 
self-adjoint operator in $\wh\H_T$ with eigenvalues (in the usual sense) $\{\l_\nu\inv\}$.

The orthogonal projection of $\wh\H_T$ onto the span of the eigenfunction $f_\nu$ is
\be\label{projection1}
E_\nu f\=\frac{\la f,f_\nu\ra}{\la f_\nu,f_\nu\ra}\,f_\nu\=\frac{(f,\s f_\nu)}
{\la f_\nu,f_\nu\ra}\,f_\nu.
\ee
Abusing notation, we write $E_\nu$ also for the kernel of the operator \eqref{projection1}:
\bea
E_\nu f(x)&=&\intv E_\nu(x,y)f(y)\,dy,\qquad f\in \H_T\nonumber\\
E_\nu(x,y)&=&\frac1{\la f_\nu,f_\nu\ra}
\bsm f_{\nu,1}(x)f_{\nu,2}(y)&
f_{\nu,1}(x)f_{\nu,1}(y)\\ f_{\nu,2}(x)f_{\nu,2}(y)&f_{\nu,2}(x)f_{\nu,1}(y)\esm. 
\label{projection2}
\eea
We may also view the operator $E_\nu$ as a projection of $\H$ onto the span of $f_\nu$.
We view the kernels $E_\nu$ as belonging to the Hilbert space 
\bes
\H^{(2)}\= L^2(I\times I; M(2,\C)) 
\ees
of mappings from the square $I\times I=[-1,1]\times [-1,1]$ to the space $M(2,\C)$ of
complex $2\times 2$ matrices.

Let us write the integral kernel for the operator $D_1\inv$ as $\kappa$.
Since $\frac12|y\pm 1|\le 1$ for all $y\in[-1,1]$, we have $|\kappa|\le 1$.

If $r$ is a bounded positive measure on the interval $[-1,1]$, then
\bes
|D_1\inv r(x)|\=\left|\int _{-1}^1 \kappa(x,y)\,dr(y)\right|.
\ees

\begin{lemma}\label{key0} Suppose that $r$ is a bounded positive measure and $\psi$ is a 
continuous function on the interval $[-1,1]$.  Then
\be\label{key1}
|[D_1\inv \psi r](x)|^2\ \le \ \overline r\intv \psi(y)^2\,dr(y), \quad \overline r=\intv dr(y). 
\ee\end{lemma}

\proof  This follows by applying the Cauchy--Schwarz inequality to 
the term on the right in the inequality
\bes
\left|[D_1\inv (\psi r)](x)\right|\=\left|\intv\kappa(x,y)\psi(y)\,dr(y)\right|
\ \le \ \intv |\psi(y)|\,dr(y).\qquad\qed
\ees

\begin{proposition}\label{key2} Each element of the kernel \eqref{projection2} has absolute
value bounded by $(\overline m+\overline n)|\l_\nu|$, where $\overline m=\intv dm$ and 
$\overline n=\intv dn$.
\end{proposition}

\proof Since $f_\nu=\l_\nu Tf_\nu$, we have 
\bes
f_{\nu,1}\=\l_\nu D_1\inv (\vp_{\nu,2}n),\quad \vp_\nu\=D_0\inv f_\nu.  
\ees
By Lemma \ref{key0},
\bes
|f_{\nu,1}(x)|^2\ \le\\\l_\nu^2 \overline n\intv \vp_{\nu,2}^2\,dn\ \le \ \overline n|\l_\nu \la f_\nu,f_\nu\ra|.   
\ees
Similarly,
\bes
|f_{\nu,2}(x)|^2\ \le\ \overline m|\l_\nu\la f_\nu,f_\nu\ra|.
\ees
Therefore each entry of $E_\nu$ is bounded by one of $|\l_\nu|\overline m$, $|\l_\nu|\overline n$, 
or $|\l_\nu|(\overline m \overline n)^{1/2}$.
\qquad$\qed$

\mds
For later use we introduce the operator $R_\l$ defined by
\be\label{Rl}
R_\l\=I-(I-\l T)\inv\=-\l T(1-\l T)\inv.
\ee
The operator $R_\l$ maps $\H$ to $\H_T$ and is compact, as is the extension to
$\wh H_T$ of its restriction to $\H_T$.  Since $Tf_\nu=\l_\nu\inv f_\nu$,
\bes
R_\l f_\nu\=-\frac{\l}{\l_\nu}\left(1-\frac{\l}{\l_\nu}\right)\inv f_\nu\=
\frac{\l}{\l-\l_\nu}\,f_\nu.
\ees 
Therefore, we have a formal expansion
\be
R_\l\=\sum_\nu\frac{\l}{\l-\l_\nu}E_\nu
\ee
with a formal kernel  
\bea
\wh K_\l(x,y)&=&\sum_\nu\frac{\l}{\l-\l_\nu}E_\nu(x,y)\nonumber\\
&=&\sum_\nu\frac{\l}{\l-\l_\nu}\frac1{\la f_\nu,f_\nu\ra}
\bsm f_{\nu,1}(x)f_{\nu,2}(y)&
f_{\nu,1}(x)f_{\nu,1}(y)\\ f_{\nu,2}(x)f_{\nu,2}(y)&f_{\nu,2}(x)f_{\nu,1}(y)\esm. 
\label{Kl1}
\eea
The question is: does this series converge, in some sense, to the kernel of $R_\l$?

\begin{theorem}\label{weak-convergence} 
For each $\l$ that is not in the set of eigenvalues $\{\l_\nu\}$,  
the partial sums of the series on the right in \eqref{Kl1}
converge weakly in $\H^{(2)}$.  The weak limit $\wh K_\l$ is the kernel for $R_\l$.
\end{theorem}

\proof For $|\l_\nu|>2|\l$|, $|(\l-\l_\nu)\inv|$ is less than $2/|\l_\nu|$.  By Proposition \ref{key2},
the corresponding summand is $O(1)$ as an element of $\H^{(2)}$.  Linear combinations of
matrix functions $f(x)g(y)^t$, $f,g\in \H$ are dense in $\H^{(2)}$.   If either $f$ or $g$ is in
$N+{\rm span}\{f_\nu\}$, integration against \eqref{Kl1} yields a finite sum which is 
$( f,R_\l g)$.  This 
proves the weak convergence. Let $\wh K_\l$ be the weak limit.  As an element of
$\H^{(2)}$, it induces a bounded operator in $\H$.  This operator agrees with $R_\l$
on a dense subspace, so it is $R_\l$.\qquad$\qed$

\mds
In the next section we derive a closed-form expression for the kernel $\wh K_\l$.

\sms
Theorem \ref{weak-convergence} can be strengthened considerably under a strengthening
of the assumption \eqref{support} (b), that $m$ and $n$ have the same support:
namely that each is dominated by the other.  This can be put in the form
\be\label{support2}
{\rm (b')}  \hbox{ There is a constant $C$ such that $m+n\le Cm$ and  $m+n\le Cn$.}
\ee

\begin{lemma} Under assumption \eqref{support2},
for each eigenfunction $f_\nu$, 
\be\label{comparison}
|\la f_\nu,f_\nu\ra|\ \le\ (f_\nu,f_\nu)\ \le\ C|\la f_\nu,f_\nu\ra|. 
\ee\end{lemma}

\proof The first inequality is \eqref{norms}.
To prove the second inequality, let $\vp_\nu=D_0\inv f_\nu$.  Then
\be\label{fnu1}
(f_\nu,f_\nu)\=(D\vp_\nu,D\vp_\nu)\=-(D^2\vp_\nu,\vp_\nu)\=-\l_\nu(M\vp_\nu,\vp_\nu)
\ee
and
\bea
(M\vp_\nu,\vp_\nu)^2&=&\left[\intv \vp_{\nu,1}\vp_{\nu,2}d(m+n)\right]^2\nonumber\\
&\le&\intv\vp_{\nu,1}^2d(m+n)\intv\vp_{\nu,2}^2d(m+n)\nonumber\\
&\le& C^2\intv \vp_{\nu,1}^2\,dm\intv \vp_{\nu,2}^2\,dn\nonumber\\
&\le& C^2\left\{\intv[\vp_{\nu,1}^2\,dm+\vp_{\nu,2}^2\,dn]\right\}^2\nonumber\\
&=& C^2\l_\nu^{-2}\la f_\nu,f_\nu\ra^2.\label{fnu2}
\eea
Together, \eqref{fnu1} and \eqref{fnu2} establish the second inequality in 
\eqref{comparison}.  \qquad$\qed$

\mds This result leads to the following strengthening of the previous convergence
result.

\begin{theorem}\label{strong-convergence} Under assumption \eqref{support2}, for each $\l$
that is not an eigenvalue, the series \eqref{Kl1}
converges in $L^2$ norm to the kernel of $R_\l$.
\end{theorem}

\begin{remark} The preceding arguments can easily be extended to similar series.
The formal series
\bes
\sum_\nu \frac1{\l_\nu}E_\nu(x,y)
\ees
converges weakly to the kernel of $T$, and under the assumption \eqref{support2} it
converges in $L^2$ norm.  Also, under assumption \eqref{support2} the formal series
\bes
\sum_\nu E_\nu(x,y)
\ees
converges weakly to the kernel of the orthogonal projection of $\H$ onto $\H_T$.
\end{remark}

\section{Wronskians and Green's kernels} \label{sec:Wronskians} 
Let $\Phi(x,\l)$ be the partial fundamental matrix solution of \eqref{beam-system},
normalized at $x=-1$, as constructed earlier.  Let $\Psi(x,\l)$ be the matrix solution
normalized at $x=1$:
\beas
D^2\Phi&=&\l M\Phi, \quad \Phi(-1,\l)\=0, \quad D\Phi(-1,\l)\=\bo;\\
D^2\Psi&=&\l M\Psi, \quad \Psi(1,\l)\=0, \quad D\Psi(1,\l)\=-\bo.
\eeas
As for $\Phi$, condition \eqref{support} (a) implies that $\Psi$ and $D_x\Psi$ are 
continuous at $x=\pm1$.

If $C$ is a matrix, let 
\be\label{wh}
\wh C\=C^t\s.
\ee
Differentiating shows that quasi-Wronskians like $\wh\Phi_x \Psi-\wh\Phi\Psi_x$ are constant;
for example
\beas
D[\wh\Phi_x\Psi-\wh\Phi\Psi_x]&=&(\l\Phi^tM^t\s\Psi+\Phi_x^t\s\Psi_x)-
(\Phi_x^t\s\Psi_x+\l\Phi^t\s M\Psi)\\
&=&0,
\eeas
since $M^t\s=\s M$.  The value of the constant can be computed by taking $x=\pm 1$.
Considering the various possibilities, we have (taking $\l$ as given, $\l\in(\C\setminus\R))$:
\bea
\wh\Phi_x\Phi-\wh\Phi\Phi_x&=&0,\quad {\rm so}\ \ \wh\Phi\inv\wh\Phi_x\=\Phi_x\Phi\inv;
\label{phiphi}\\
\wh\Psi_x\Psi-\wh\Psi\Psi_x&=&0,\quad {\rm so}\ \ \wh\Psi\inv\wh\Psi_x\=\Psi_x\Psi\inv;
\label{psipsi}\\
\wh\Phi_x\Psi-\wh\Phi\Psi_x&=&-C_-\=\s\Psi(-1,\l)\=\wh\Phi(1,\l);
\label{phipsi}\\
\wh\Psi_x\Phi-\wh\Psi\Phi_x&=&C_+\=-\s\Phi(1,\l)\=-\wh\Psi(-1,\l).
\label{psiphi}
\eea
Combining some of these identities we find that
\bea
C_-&=&-\wh\Phi[\wh\Phi\inv\wh\Phi_x-\Psi_x\Psi\inv]\Psi\=\wh\Phi A\Psi;\label{Cminus}\\
C_+&=&\wh\Psi[\wh\Psi\inv\wh\Psi_x-\Phi_x\Phi\inv]\Phi\=\wh\Psi A\Phi;\label{Cplus}
\eea
where
\be\label{A}
A(x,\l)\=\Psi_x\Psi\inv-\Phi_x\Phi\inv.
\ee
The identities \eqref{phipsi} and  \eqref{Cplus} lead to two additional important identities:
\be\label{zero}
\Psi C_-\inv\wh\Phi-\Phi C_+\inv\wh \Psi\=A\inv-A\inv\=0
\ee
and
\be\label{Id}
\Psi_x C_-\inv\wh\Phi-\Phi_x C_+\inv \wh\Psi\=\Psi_x\Psi\inv A\inv-\Phi_x\Phi\inv A\inv
\=\bo.
\ee

The identities \eqref{phiphi}, \eqref{psipsi}, \eqref{Cminus}, \eqref{Cplus}, 
\eqref{A}, \eqref{zero} and \eqref{Id} call for some discussion.

\begin{lemma} Suppose $\l$ is not real.  Then  the matrix functions\ $\Phi(x,\l)$ and\ $\Psi(x,\l)$ are invertible
for $x\in (-1,1]$, $x\in [-1,1)$, respectively. \end{lemma}

\proof Suppose that $\Phi(x,\l)$ is not invertible at some point $x_0$ in the interval $(-1,1)$.
Then $x_0>-1$ and $\l$ is a Dirichlet eigenvalue for the beam problem restricted to the
interval $[-1,x_0]$.  Therefore $\l$ is real.  The same argument applies to $\Psi$ for
an interval $[x_1,1]$.\qquad$\qed$\endproof

It follows that all the expressions above are well-defined when $\l$ is not real and $\abs{x}<1$. 
Moreover, \eqref{phipsi} and \eqref{psiphi} imply that $C_-$ and $C_+$ are invertible
if $\l$ is not real.  In turn, these imply that $A$ is invertible if $\l$ is not
real and $\abs{x}<1$.  In summary,

\begin{corollary} The identities \eqref{phiphi} -- \eqref{Id} are valid for each non-real 
$\l$ and $\abs{x}<1$.  \end{corollary}  

\begin{theorem}\label{Gl} The matrix function
\be\label{G}
G_\l(x,y)\=\begin{cases} \Psi(x,\l)\,C_-\inv\,\wh\Phi(y,\l),\quad y<x;\\
\Phi(x,\l)\,C_+\inv\,\wh\Psi(y,\l),\quad y>x.\end{cases}
\ee
is the Green's kernel for the equation $D^2u-\l Mu=f$, $\l\notin\R$.\end{theorem}

\proof Let
\bea
u(x)&=&\intv G_\l(x,y)f(y)\,dy\nonumber\\
&=&\Psi(x)\pint^x C_-\inv\wh\Phi(y)f(y)\,dy
+\Phi(x)\int_x^1 C_+\inv\wh\Psi(y)f(y)\,dy.\label{Gf}
\eea
Then, using \eqref{zero},
\beas
Du(x)&=&[\Psi(x)C_-\inv\wh\Phi(x)-\Phi(x)C_+\inv\wh\Psi(x)]f(x)\\
&&\quad +\Psi_x(x)\pint^x C_-\inv\wh\Phi(y)f(y)\,dy+\Phi_x(x)\int_x^1 C_+\inv\wh\Psi(y)f(y)\,dy\\
&=& 0+\Psi_x(x)\pint^x C_-\inv\wh\Phi(y) f(y)\,dy+\Phi_x(x)\int_x^1 C_+\inv\wh\Psi(y) f(y)\,dy.
\eeas
Therefore, using \eqref{Id},
\beas
D^2u(x)&=&\l M(x)u(x)+[\Psi_x C_-\inv\wh\Phi-\Phi_x C_+\inv\wh\Psi](x)f(x)\\
&=& \l M(x)u(x)+f(x).\qquad\qed
\eeas

\mds
The kernel $G_\l$ can be used to calculate the kernel for the operator
$R_\l$ defined in  \eqref{Rl}.  Given $f\in \H_1$ with $g=Df$ integrable,
consider the inhomogeneous problem
\be\label{inhom}
(D^2-\l M)u\=g
\ee
with solution 
\bes
u(x)\=\intv G_\l(x,y)g(y)\,dy.
\ees
Let $v=Du$, so \eqref{inhom} is equivalent to
\bes
Dv-\l MD_0\inv v\=Df 
\ees
or $v-\l Tv=f$, which is the same as 
\be\label{inhom2}
v \=f-R_\l f.
\ee
Now
\bes
u(x)\=\intv G_\l(x,y)g(y)\,dy.
\ees
Thus the solution to \eqref{inhom2} is, using \eqref{zero} and \eqref{Id} again, 
\beas
v(x)&=&Du(x)\\
&=&\Psi_x(x)\pint^x C_-\inv \wh\Phi(y)Df(y)\,dy+\Phi_x(x)\int_x^1C_+\inv
\Phi(y)Df(y)\,dy\\
&=& f(x)-\Psi_x(x)\pint^x C_-\inv\wh\Phi_y(y)f(y)\,dy-\Phi_x(x)\int_x^1 C_+\inv
\wh \Psi_y (y)f(y)\,dy.
\eeas
This shows that the kernel 
\be\label{Gxy}
[G_\l]_{xy}\=
\begin{cases} \Psi_x(x)\,C_-\inv \wh\Phi_y(y),\qquad y<x;\\
\Phi_x(x)\,C_+\inv\, \wh\Psi_y(y),\qquad y>x.\end{cases}
\ee
generates $R_\l$ for functions in $\H_1$.  The operator $R_\l$ is zero on
$\H_0$, while the integral of $[G_\l]_{xy}$ against $\bo$ is
\bes
\intv [G(x,y)_\l]_{xy}\,dy
\=\Psi_x(x)\pint^x C_-\inv \wh\Phi_y(y)\,dy
+\Phi_x(x)\int_x^1C_+\inv \wh\Psi_y(y)\,dy,
\ees 
which, by \eqref{Id}, is the identity matrix. We can compensate 
by subtracting $\frac12\bo$ from  $[G_\l]_{xy}(x,y)$ as kernel.

Summarizing, 

\begin{theorem}\label{Kl3} The kernel for the operator $R_\l=-\l T(I-\l T)\inv$ is
\be\label{Kl4}
K_\l(x,y)\=[G_\l]_{xy}(x,y)-\tfrac12\bo.
\ee\end{theorem}

We can now relate this to the kernel $\wh K_\l$ defined using the formal series
\eqref{Kl1}.  We have shown that they  each define the same operator $R_\l$
in $\H$.  This means that they coincide as elements of the $L^2$ space $\H^{(2)}$ [*],
so we may choose to identify them at each point $(x,y)$.

\begin{theorem}\label{sameKl} The kernels $K_\l$ and $\wh K_\l$ are identical.\end{theorem}

\mds
For later use we define here the {\it Weyl function\/} for the beam Dirichlet problem
to be 
\be\label{weyl-green}
W(\l)\=\frac1\l D\Phi(1,\l)\Phi(1,\l)\inv,
\ee
The representation \eqref{Kl4} shows that
\be\label{weyl-def}
W(\l)\=\frac1\l K_\l(1,1)+\frac{1}{2\l} \bo.
\ee
It follows that $W$ has a pole at the origin with residue 
\be\label{Czero}
D\Phi(1,0)\Phi(1,0)\inv\=\frac12\bo.
\ee
Then the representation \eqref{Kl1} shows that, formally at least,
\be\label{weyl-series}
W(\l)\=\frac1{2\l}\bo+\sum_\nu \frac1{\l-\l_\nu}\frac1{\la f_\nu,f_\nu\ra}
\bsm f_{\nu,1}(1)f_{\nu,2}(1)&
f_{\nu,1}(1)f_{\nu,1}(1)\\ f_{\nu,2}(1)f_{\nu,2}(1)&f_{\nu,2}(1)f_{\nu,1}(1)\esm. 
\ee
(We omit a detailed justification of \eqref{weyl-series} in the general case, since the
only use we shall make is to the case when $\{\l_\nu\}$ is finite.)

We assume here that the $f_\nu$ are chosen to be real.
It will be useful to understand the signs of the entries of the summands. 

\begin{lemma}\label{fnu-signs} 
If\ $\nu>0$, then $f_{\nu,1}$ and  $f_{\nu,2}$ have the same sign.\end{lemma}

\proof Because of the relationship between eigenfunctions for $\pm\l_\nu$, 
this is equivalent to the statement that if $\nu<0$ then $f_{\nu,1}$ and
$f_{\nu,2}$ have opposite signs.  With our choice of indexing, $\nu<0$ means
$\l_\nu>0$.  The  corresponding $\phi_\nu$ is $\Psi v$ for some fixed 2-vector $v$.
Then  $\Psi(-1,\l_\nu)v=\bz$.  But $\l_\nu>0$ implies that all entries of $\Psi(-1,\l_\nu)$
are positive (this is the dual of the argument for Proposition \ref{Phik}) so $v_1v_2<0$.  Then
$D\phi_\nu(1)=\Psi_x(1,\l)v=-v$. Since $f_\nu$ is a multiple of $v$, its entries
have opposite signs.\qquad$\qed$

To simplify the notation in \eqref{weyl-series}, let 
\be
\a_\nu\=\frac{f_{\nu,1}(1)}{\sqrt{\la f_\nu,f_\nu\ra}},\quad
\b_\nu\=\frac{f_{\nu,2}(1)}{\sqrt{\la f_\nu,f_\nu\ra}},\quad \hbox{for\ \  $\nu>0$}.
\ee
By Lemma \ref{fnu-signs} we may take $\a_\nu$ and $\b_\nu$ positive.
Taking into account the relation between $f_\nu$ and $f_{-\nu}$ and between
$\la f_\nu,f_\nu\ra$ and $\la f_{-\nu},f_{-\nu}\ra$  it follows that 
\bea
W(\l)&=&\frac1{2\l}\bo+\sum_{\nu<0}\frac1{\l-\l_\nu}\bsm \a_\nu \b_\nu& -\a_\nu^2
\\ -\b_\nu^2& \a_\nu \b_\nu\esm\nonumber\\
&&\quad +\sum_{\nu>0}\frac1{\l-\l_\nu}\bsm \a_\nu \b_\nu& \a_\nu^2\\ \b_\nu^2& \a_\nu \b_\nu\esm, 
\label{weyl-series2}
\eea
where we set $\a_{-\nu}=\a_\nu, \, \b_{-\nu}=\b_\nu$.  

\section{The discrete beam} \label{sec:discrete beam}
The discrete beam is characterized by measures $m$ and $n$ that are supported on discrete
points
\bes
-1\ <\ x_1\ <\ x_2\ <\ \dots\ <\ x_{d-1}\ <\ x_d\ <\ 1,
\ees
with masses $m_j$, $n_j$.  For convenience we also define
\bes
x_0\=-1,\quad\ \  x_{d+1}\=1,\quad\ \  l_j\=x_{j+1}-x_j, \quad\ \  M_0\=\bz.
\ees
Here conditions \eqref{support} and \eqref{support2} both 
reduce to the assumption that $m_jn_j>0$, $j=1,\dots d$.  

The partial fundamental 
solution $\Phi(x,\l)$ satisfies $D^2 \Phi=0$ except at the $x_j$, so it is piecewise
linear in $x$, and the derivative $D\Phi$ is piecewise constant.  Thus for any given
$\l$ the function $\Phi$ is characterized by its values 
\be\label{Phij}
\Phi_j\=\Phi_j(\l)\=\Phi(x_j,\l),\quad j=0,\dots d+1.
\ee
Similarly, $D\Phi=\Phi_j'$ is characterized by its one-sided values
\be\label{DPhij}
\Phi_j'\=D\Phi_j(\l)\=D\Phi(x_j-,\l),\quad j=1,\dots d+1.
\ee
The beam equation $D^2\Phi=\l M\Phi$, with initial conditions 
\bes \Phi(-1,\l)\=\bz, \qquad D\Phi(-1,\l)=\bo,
\ees
 translates to the conditions
\beas
\Phi_0&=&\bz,\qquad \Phi_{j+1}\=\Phi_j+l_j\Phi'_{j+1},\\
\Phi'_1&=&\bo,\qquad \Phi'_{j+1}\=\Phi'_j+\l M_j\Phi_j.
\eeas
These relations can be put in two forms:
\be\label{Lj}
\bsm \Phi_{j+1}\\ \Phi'_{j+1}\esm\=\bsm \bo &  l_j\bo\\ \bz&\bo\esm
\bsm \Phi_{j}\\ \Phi'_{j+1}\esm 
\ee
and
\be\label{Tj}
\bsm \Phi_{j+1}\\ \Phi'_{j+1}\esm\=\bsm \bo+\l l_j M_j& l_j\bo \\ 
\l M_j&\bo\esm
\bsm \Phi_{j}\\ \Phi'_{j}\esm\=T_j\bsm \Phi_{j}\\ \Phi'_{j}\esm.
\ee

\begin{lemma}\label{Phij-degree} Each of $\Phi_j$ and $\Phi'_j$ is a polynomial of degree
$j-1$; the even part is  diagonal and the odd part is off-diagonal. The Dirichlet spectrum
has $2d$ elements. \end{lemma}

\proof The first statement follows by induction from the recursion relations \eqref{Tj}.
A consequence is that the determinant $\D(\l)=\det\Phi_{d+1}$ is a
polynomial of degree $d$ in $\l^2$, so the eigenvalues come in $d$ pairs. 
\qquad$\qed$

\mds
As shown in the general case, the eigenvalues are distinct and real. 

\sms
The Weyl function \eqref{weyl-green} in the discrete case is
\be\label{weyl-discrete}
W(\l)\=\frac1\l \Phi'_{d+1}\Phi_{d+1}\inv.
\ee
The recursion relations \eqref{Lj} 
imply that $W(\l)$  has a continued fraction expansion involving 
non-commuting coefficients \cite{stieltjes, wedderburn}.  

\begin{proposition} \label{prop:WStieltjes} 
\mbox{}
\begin{equation}\label{partial-frac}
 W(\lambda) =\cfrac{1}{\l l_d \mathbf{1}+\cfrac{1}{M_d+\cfrac{1}{\l l_{d-1} \mathbf{1}+
\cfrac{1}{M_{d-1}+\cfrac{1}{\ddots+\cfrac{1}{\l l_0\mathbf{1}}}}}}}
\end{equation} 
\end{proposition} 
\begin{proof} 
Let $W_j=\lambda\inv \Phi'_j\Phi_j^{-1}$.  
Note that, for $\lambda$ large enough, all  $\Phi'_j$ and $\Phi_j$ are invertible.  
The relations
\bes
\Phi_{d+1}\=\Phi_d+l_d \Phi'_{d+1},\qquad
\Phi'_{d+1}=\Phi'_d+\lambda M_d\Phi_d, 
\ees
imply that
\bes
\Phi_{d+1}(\Phi'_{d+1})^{-1}=(\lambda M_d +\Phi'_d\Phi_d^{-1})^{-1} +l_d \mathbf{1}, 
\ees
hence 
\bes
W_{d+1}^{-1}=\l l_d \mathbf{1} +(M_d+W_d)^{-1}.  
\ees
Inverting this expression we obtain
\bes
W_{d+1}=[\l l_d \mathbf{1} +(M_d+W_{d} )^{-1}]^{-1}.
\ees
Iterating down to $W_1=(\l l_0 \mathbf{1})^{-1}$ concludes the proof.  
\qquad \end{proof} 

We want to reverse this procedure and recover the data $\{l_j\}$ and $\{M_j\}$ from the
function $W$.  We follow the procedure of Stieltjes \cite{stieltjes}, starting with the
determination of certain Pad\'e approximants of $W$.  At step zero, let
\bes
P_0\=0,\quad Q_0\=\bo,
\ees
so 
\be\label{Pade0}
Q_0 W\=P_0+O(\l\inv);\qquad W\=Q_0\inv P_0+O(\l\inv).
\ee
To proceed, we note that 
\be\label{Tj-inverse}
T_j\inv\=\bsm \bo &-l_j\bo\\ -\l M_j& \bo+\l l_j M_j\esm.
\ee
Therefore the identity \eqref{Tj} implies that
\bes
T_d\inv \bsm \Phi_{d+1}\\ \Phi'_{d+1}\esm\=
\bsm \Phi_{d+1}-l_d \Phi'_{d+1}\\ -\l M_d\Phi_{d+1}+(\bo+\l l_d M_d)\Phi'_{d+1} \esm
\=\bsm\Phi_d\\ \Phi'_d\esm.
\ees
Multiplying each (block) row on the right by $\Phi_{d+1}\inv$, we obtain
\bes
\bsm \bo-l_d(\l W)\\ -\l M_d+(1+\l l_dM_d)\l W\esm\=\bsm \Phi_d\Phi_{d+1}\inv\\
\Phi'_d\Phi_{d+1}\inv\esm.
\ees
The two equations for $W$ can be rewritten as
\bea
\l l_d W&=&\bo-\Phi_d\Phi_{d+1}\inv\=\bo+O(\l\inv);\label{Pade1}\\
(\bo+\l l_dM_d)W&=&M_d-\l\inv \Phi'_d\Phi_{d+1}\inv\=M_d+O(\l^{-2})\label{Pade2}.
\eea
Set 
\be\label{PQ12}
P_1=\bo,\quad Q_1\=\l l_d\bo;\qquad P_2\=M_d,\quad Q_2\=\l l_dM_d+\bo.
\ee
Then $Q_1\inv P_1$ and $Q_2\inv P_2$ are Pad\'e approximants to $W$ on the left:

\be\label{Pade12}
W\=Q_1\inv P_1+O(\l^{-2}),\qquad W\=Q_2\inv P_2+O(\l^{-3}).
\ee
These two approximates are uniquely determined by the conditions $P_1(0)=\bo$, 
$Q_2(0)=\bo$, respectively; see the next section.

This process can be continued.  We have
\be\label{fromd-j+1}
[T_dT_{d-1}\cdots T_{d-j+1}]\inv \bsm \Phi_{d+1}\\ \Phi'_{d+1}\esm\=\bsm \Phi_{d-j+1}
\\ \Phi'_{d-j+1}\esm.
\ee
Let us write, in a temporary notation for this section only,
\be\label{fromd-j+1b}
[T_d\dots T_{d-j+1}]\inv\=\bsm a_j(\l)&-b_j(\l)\\ -c_j(\l)& d_j(\l)\esm, \quad 1\le j\le d.
\ee

\begin{lemma}\label{abc}  
(a)\ The polynomials $a_j$ and $b_j$ have degree $j-1$; the polynomials
$c_j$ and $d_j$ have degree $j$, $1\le j\le d$.\end{lemma}

\sms\nin (b)\ For each $1\le j\le d$, $a_j(0)=d_j(0)\=\bo$,\ $b_j(0)=c_j(0)=\bz$.

\sms\nin(c)\ The coefficients of even powers in $a_j$,$b_j$, $c_j$ and $d_j$ are
diagonal and the coefficients of odd powers are off-diagonal.

\proof Note that each of these statements is true at $j=1$:
\bes
\bsm a_1 & -b_1\\ -c_1& d_1\esm\=T_d\inv\=\bsm \bo&-l_d\bo\\ -\l M_d& \bo+\l l_d M_d.\esm
\ees
Note that
\bea
&&\bsm a_{j+1}&-b_{j+1}\\ -c_{j+1}& d_{j+1}\esm\= T_{d-j}\inv
\bsm a_j&-b_j\\ -c_j& d_j\esm \label{transition}\\
&&=\bsm a_j+l_{d-j}c_j & -b_j-l_{d-j}d_j\\ -
\l M_{d-j} a_j-(\bo+\l l_{d-j}M_{d-j})c_j&
\l M_{d-j}b_j+(\bo+\l l_{d-j}M_{d-j})d_j\esm.\nonumber
\eea
The assertion (a) follows by easily by induction.  Each $T_j(0)\=\bo$, which implies (b).
Assertion (c) follows from the fact that 
multiplication by any entry of $T_j\inv$ preserves these properties.\qquad$\qed$

\mds
In analogy with the computations that led to \eqref{Pade1} and \eqref{Pade2}, we
multiply each (block) row of the identity
\bes
\bsm \Phi_{d-j+1}\\ \Phi'_{d-j+1}\esm\=\bsm a_j&-b_j\\ -c_j&d_j\esm\bsm\Phi_{d+1}\\ \Phi'_{d+1}\esm
\ees
on the right by $\Phi_{d+1}\inv$ and obtain the equations
\beas
\l b_j W&=&a_{j}-\Phi_{d-j+1}\Phi_{d+1}\inv;\\
d_{j}W&=& \l\inv c_{j}+\l\inv\Phi'_{d-j+1}\Phi_{d+1}\inv.
\eeas
Accordingly, and consistent with previous definitions for $j=1$,
\bea
P_{2j-1}&=&a_j,\qquad  Q_{2j-1}=\l b_j,\quad 1\le j\le d;\label{PQ2j-1}\\
P_{2j}&=&\l\inv c_j,\qquad Q_{2j}\=d_j,\quad 0\le j\le d.\label{PQ2j}
\eea
Note that since $c_j(0)=0$, each of the $P_k$, $Q_k$ is a polynomial.

In view of \eqref{Pade0} \eqref{Pade1}, \eqref{Pade2}, and Lemma \ref{abc}, we have

\begin{proposition}\label{Pade-properties} 
The polynomials $P_k$, $Q_k$, $1\le k\le 2d$, have the properties

\sms\nin (a) $Q_{2j-1}$ and $Q_{2j}$ have degree $j$, $P_{2j-1}$ and $P_{2j}$ have degree  
$j-1$;

\sms\nin (b) The coefficient of odd powers of $Q_{2j-1}$ are diagonal, and the coefficients
of even powers are off-diagonal;

\sms\nin (c) The coefficient of even powers of $Q_{2j}$ are diagonal, and the coefficients
of odd powers are off-diagonal.

\sms Moreover
\bea
Q_{2j-1}W&=&P_{2j-1}+O(\l^{-j});\label{WQodd}\\
Q_{2j}W&=&P_{2j}+O(\l^{-j-1}).\label{WQeven}
\eea
\end{proposition}

\mds
In the next section we treat the {\it inverse problem\/}: the problem of recovering the
beam data $\{l_j\}$, $\{M_j\}$ from $W$.  The final step of the process described there
uses the fact that the data can be recovered from the leading coefficients of the 
polynomials $\{Q_k\}$.

\begin{proposition}\label{leading} Let $\la Q_k\ra$ denote the leading coefficient of
$Q_k$.  Then for $1\le j\le d$,
\bea
\la Q_{2j-1}\ra \la Q_{2j-2}\ra\inv&=& l_{d-j+1}\bo;\label{l-quotient}\\
\la Q_{2j}\ra \la Q_{2j-1}\ra\inv&=&M_{d-j+1}.\label{M-quotient}
\eea
\end{proposition}

\proof Let $\la a_j\ra$, $\la b_j\ra$, $\la c_j\ra$, $\la d_j\ra$ denote the leading
coefficients of $a_j$, $b_j$, $c_j$, $d_j$.  Because of Lemma \ref{abc} (a) and 
\eqref{transition}, it follows that  the recursion for the matrix of principal 
coefficients is given by
\be\label{principal-recursion}
\bsm \la a_{j}\ra & -\la b_{j}\ra\\ -\la c_{j}\ra&\la d_{j}\ra\esm
\=\bsm l_{d-j+1}\la c_{j-1}\ra& -l_{d-j+1}\la d_{j-1} \ra    \\ -l_{d-j+1}M_{d-j+1} \la c_{j-1} \ra    
& l_{d-j+1}M_{d-j+1}\la d_{j-1}  \ra\esm.
\ee
At the first step, $Q_0=\bo$ and $\la Q_1\ra=l_d\bo$, so
\bes
\la Q_1\ra\la Q_0\ra\inv\=\la Q_1\ra\=l_d\bo.
\ees
At each subsequent step, \eqref{principal-recursion} implies that
\beas
\la Q_{2j-1}\ra&=&\la b_j\ra\=l_{d-j+1}\la d_{j-1}\ra\=l_{d-j+1}\la Q_{2j-2}\ra,
\eeas
which proves \eqref{l-quotient}. Similarly, at each step \eqref{principal-recursion} 
implies that
\bes
\la Q_{2j}\ra\=\la d_j\ra\=M_{d-j+1}\la b_j\ra\=M_{d-j+1}\la Q_{2j-1}\ra,
\ees
which proves \eqref{M-quotient}. \qquad$\qed$

\section{The inverse problem for the discrete beam} \label{sec:inverse problem} 

We shall show that the Weyl function $W$ has an asymptotic expansion
\be\label{W-asymptotics} 
W(\l)\= \frac1\l\,C_0+\frac1{\l^{2}}\,C_1+\dots\ \frac1{\l^{n+1}}\,C_n
+O\left(\frac1{\l^{n+2}}\right) \quad\hbox{as\ \ $\l\to\infty$.}
\ee

The denominators  $Q_k$ of the Pad\'e approximants to $W$ can be recovered
from this asymptotic expansion of $W$.  For example, subsitute the 
expansion \eqref{W-asymptotics} for $W$ in \eqref{WQodd} and expand.
Since $Q_{2j-1}$ has no constant term and the constant term of $P_{2j-1}$
is $\bo$, the term of order $0$ in the expansion is $\bo$ and the terms of order $-1,\dots,1-j$ 
in the expansion are zero.  Writing
\be\label{Qs-odd}
Q_{2j-1}\=\l^{j}Q^{(j-)}_{j}+\dots+\l^2Q^{(j-)}_2+\l Q^{(j-)}_1,
\ee
the resulting system of equations can be written
\bea
&&\bsm Q^{(j-)}_{1} & Q^{(j-)}_2&\dots&Q^{(j-)}_{j}\esm
\bsm C_0&C_1&\dots& C_{j-1}\\ C_1&C_2&\dots&C_j\\
& &\dots& \\ C_{j-1}&C_{j}&\dots& C_{2j-2}\esm\nonumber\\
&&\quad \=\bsm \bo&0&\dots&0\esm.\label{Qodd-system}
\eea
Write
\be\label{Qs-even}
Q_{2j}\=\l^{j}Q^{(j+)}_{j}+\dots+\l Q^{(j+)}_{1}+\bo.
\ee
Since $Q_{2j}(0)=\bo$, the same argument leads to the system 
\bea
&&\bsm Q^{(j+)}_{1} & Q^{(j+)}_2&\dots&Q^{(j+)}_{j}\esm
\bsm C_1&C_2&\dots& C_{j}\\ C_2&C_3&\dots&C_{j+1}\\
& &\dots& \\ C_j&C_{j+1}&\dots &C_{2j-1}\esm\nonumber\\
&&\quad \=-\bsm C_0&C_1&\dots&C_{j-1}\esm.
\label{Qeven-system}
\eea

In principle, the matrix equations \eqref{Qodd-system} and \eqref{Qeven-system}, considered as
scalar equations, consist of $4j$ linear equations in $4j$ unknowns.  However we know 
that each of the coefficients of $Q^{(j^\pm)}$ is either a diagonal or an off-diagonal 
matrix, so there are only $2d$ unknowns.  Moreover, as we shall show, the same is true of each 
of the matrices $C_k$, so the associated $2j\times 2j$ matrix for these equations has only
$2\cdot j^2$ non-zero entries.  As we shall show, each system \eqref{Qodd-system} and
\eqref{Qeven-system} decomposes easily into two uncoupled systems of $j$ equations in $j$
unknowns, permitting simple formulas for the leading coefficients.  

To understand the $C_k$, we return to the formula \eqref{weyl-series2} for $W$: 

\bes
W(\l)\=\frac1{2\l}\bo+\sum_{\nu=-d}^{-1}\frac1{\l-\l_\nu}\bsm \a_\nu\b_\nu& -\a_\nu^2\\ 
-\b_\nu^2& \a_\nu\b_\nu\esm
+\sum_{\nu=1}^d\frac1{\l-\l_\nu}\bsm \a_\nu\b_\nu& \a_\nu^2\\ \b_\nu^2& \a_\nu\b_\nu\esm, 
\ees
where $\a_\nu=\a_{-\nu}$ and  $\b_\nu=\b_{-\nu}$ are positive.  

For large $|\l|$, $(\l-\l_\nu)\inv=\sum_{n=0}^\infty \l_\nu^k/\l^{k+1}$, so 
\be\label{C0}
C_0\=
\left[\frac12+\sum_{\nu=1}^d 2\a_\nu\b_\nu\right]\bo,
\ee
and
\bes
C_k\=\sum_{\nu=1}^d \left\{\l_{-\nu}^k\bsm \a_\nu\b_\nu& -\a_\nu^2\\ -\b_\nu^2& \a_\nu\b_\nu\esm
+\l_\nu^k\bsm \a_\nu\b_\nu& \a_\nu^2\\ \b_\nu^2& \a_\nu\b_\nu \esm\right\},\qquad k\ge 1.
\ees
Recall that $\nu$ and $\l_\nu$ have opposite signs, so
\bes
C_k\=\begin{cases}\ \ \sum_{\nu=1}^d 2|\l_\nu|^k\bsm 0& -\a_\nu^2\\ -\b_\nu^2&0\esm,     
&\ \hbox{$k$\ \ odd};\\
& \\
\ \ \sum_{\nu=1}^d 2|\l_\nu|^k\bsm \a_\nu\b_\nu &0\\ 0&\a_\nu\b_\nu\esm, &\ 
\hbox{$k$\ \  even,\ \ \ $k\ge 2$.}\end{cases}
\ees
Thus
\be\label{C_k}
C_k\=\bsm a_k&0\\0&a_k\esm,\quad k\ \ \hbox{even};\qquad
C_k\=\bsm 0&b_k\\ c_k&0\esm,\quad k\ \ \hbox{odd},
\ee
where
\beas a_0&=&\frac12+2\sum_{\nu=1}^d \a_\nu\b_\nu;\qquad a_{k}
\=2\sum_{\nu=1}^d|\l_\nu|^{k} \a_\nu\b_\nu,\quad k \  {  \rm even }, \, k\ge 2;\\ 
b_k&=&-\sum_{\nu>0}2|\l_\nu|^k\a_\nu^2;\qquad
c_k=-\sum_{\nu>0}2|\l_\nu|^k \b_\nu^2,\quad\ \ \  k\ \ {\rm odd.}
\eeas

Let us consider the systems \eqref{Qodd-system} and \eqref{Qeven-system} for $j=2$:
\be\label{2-systems}
\left[Q^{(2-)}_1\ Q^{(2-)}_2\right]\bsm C_0&C_1\\ C_1& C_2\esm\=[\bo\ \ \bz];
\qquad \left[Q^{(2+)}_1\ Q^{(2+)}_2\right]\bsm C_1&C_2\\ C_2& C_3\esm\=-[C_0\ C_1].
\ee
The key structural fact here is that each row or column consists of one diagonal matrix
and one off-diagonal matrix.  For larger values of $j$ there is a similar structure, with
diagonal matrices and off-diagonal matrices alternating.  Filling in the entries, the
first of the systems \eqref{2-systems} is
\be\label{Q2minus}
\bsm x_{11}&0&0&x_{21}\\ 0&x_{12}&x_{22}&0\esm
\bsm a_0&0&0&b_1\\ 0&a_0&c_1&0\\ 0&b_1& a_2&0\\ c_1&0&0&a_2\esm\=\bsm 1&0&0&0\\
0&1&0&0\esm,
\ee
where $x_{k1}$, $x_{k2}$ are the non-zero elements in the first and second rows
of the coefficient $Q^{(2-)}_k$, respectively. 

Because of the way that the positions of zero and non-zero elements in the
rows and columns either match or complement each other, there are cancellations.
For example, the product of the first row of the matrix on the left with the second
or third columns of the matrix on the right is zero.  Therefore the four equations
associated to the first row reduce to two, which can be written as a system
\be\label{Q2minus1}
\bsm x_{11}& x_{21}\esm\bsm a_0& b_1\\ c_1&a_2\esm \=\bsm 1& 0\esm.
\ee
Similarly, the four equations associated with the second row reduce to
\be\label{Q2minus2}
\bsm x_{12}& x_{22}\esm \bsm a_0& c_1\\ b_1&a_2\esm \=\bsm 1& 0\esm.
\ee
A second way to organize this is by a suitable permutation of rows and columns,
so that rows with the same pattern of zero entries are juxtaposed, and the same
for columns.  Then the original system of $8$ equations becomes
\be\label{Q2minus-reduction}
\bsm x_{11}& x_{21}&0&0\\0&0& x_{12}& x_{22} \esm
\bsm a_0 & b_1&0&0\\ c_1& a_2&0&0\\ 0&0& a_0&c_1\\0&0& b_1&a_2\esm
\=\bsm 1&0&0&0\\ 0&0&1&0\esm.  
\ee
The same procedure applies in general to the equations for the coefficients
$Q^{(j-)}_k$ of $Q_{2j-1}$, yielding an equivalent form in which 
the original $2j\times 2j$ matrix is reduced to a diagonal form
with two $j\times j$ matrices, adjoints of each other, on the diagonal. We write this
explicitly below.

\sms
A similar analysis of the second of the systems \eqref{2-systems} yields a
different form of canonical reduction.  Here the system has the form
\be\label{Q2plus} \bsm 0&x_{11}&x_{21}&0\\ x_{12}&0&0&x_{22}\esm 
\bsm 0&b_1&a_2&0\\ c_1&0&0&a_2\\ a_2&0&0&b_3\\ 0&a_2&c_3&0\esm\= 
-\bsm a_0&0&0&b_1\\ 0&a_0&c_1&0\esm
\ee 
where $x_{k1}$ and $x_{k2}$ are the non-zero entries of the first and
second rows of the coefficient $Q^{(2+)}_k$, respectively. Again the
positioning of the zeros in the rows and columns tells us that these
equations reduce to two uncoupled systems
\bea
\bsm x_{11}&x_{21}\esm \bsm c_1&a_2\\ a_2&b_3\esm&=& -\bsm a_0&b_1\esm;\label{Q2plus1}\\
\bsm x_{12}&x_{22}\esm \bsm b_1& a_2\\ a_2&c_3\esm&=&-\bsm a_0&c_1\esm.\label{Q2plus2}  
\eea
As in the case of \eqref{Q2minus1}, \eqref{Q2minus2}, the system \eqref{Q2plus} can 
be rearranged to the form
\be\label{Q2plus-reduction}
\bsm x_{11}&x_{21}&0&0\\ 0&0&x_{12}&x_{22}\esm 
\bsm 0&0&c_1&a_2\\0&0& a_2&b_3\\ b_1&a_2&0&0\\ a_2&c_3&0&0\esm\=
 -\bsm 0&0&a_0&b_1\\ a_0&c_1&0&0\esm.
\ee

Let us pass to the general case for the coefficients $Q_k^{(2j\pm)}$ of $Q_{2j-1}$
and $Q_{2j}$.
We start with the $(2d+2)\times (2d+2)$ Hankel matrix
\be\label{Hankel}
H\=\bsm C_0&C_1&C_2&\dots&C_{d}\\ C_1&C_2&C_3&\dots &C_{d+1}\\
C_2&C_3&C_4&\dots &C_{d+2}\\
&&&\ddots& \\
C_{d}&C_{d+1}&C_{d+2}&\dots& C_{2d}\esm.
\ee
Writing out the $2\times 2$ blocks,
\be\label{Hankelb}
H\=\bsm a_0&0&0&b_1&a_2&0&0&b_3\ \ \dots\ \\
        0&a_0&c_1&0&0&a_2&c_3&0\ \ \dots\ \\
        0&b_1&a_2&0&0&b_3&a_4&0\ \ \dots\ \\
        c_1&0&0&a_2&c_3&0&0&a_4\ \dots\ \\
        a_2&0&0&b_3&a_4&0&0&b_5 \ \dots\ \\
        0&a_2&c_3&0&0&a_4&c_5&0\ \ \dots\ \\
        0&b_3&a_4&0&0&b_5&a_6&0\ \dots\ \\
        c_3&0&0&a_4&c_5&0&0&c_7\ \dots\ \\      

         &&\ddots&&&\ddots&& \esm
\ee
In a notation that is best explained by \eqref{Q2minus-reduction}
we introduce  two $(d+1)\times (d+1)$ matrices constructed by reorganizing $H$: 
\bes
H^{NW}\=\bsm a_0&b_1&a_2&b_3\ \  \dots \\ 
              c_1&a_2&c_3&a_4\ \ \dots \\ 
              a_2&b_3&a_4&b_5\ \ \dots \\ 
              c_3&a_4&c_5&a_6\ \ \dots \\ 
              &&\ddots&\esm;\quad\ \
H^{SE}\=\bsm a_0&c_1&a_2&c_3\ \ \ \dots \\
             b_1&a_2&b_3&a_4\ \ \dots \\
             a_2&c_3&a_4&c_5\ \ \dots \\
             b_3&a_4&b_5&a_6\ \ \dots \\
             &&\ddots&\esm.
\ees
We denote the $j\times j$ principal minors of $H^{NW}$ and $H^{SE}$ by $H^{NW}_j$ and
$H^{SE}_j$, respectively.  Note that they are transposes of each other: 
\bes
[H^{NW}_j]^t \= H^{SE}_j.
\ees
Therefore they have the same determinant 
\be\label{HNEt}
\det H^{NW}_j\=\det H^{SE}_j\=\D_j.
\ee
Following the same procedure as for $Q_3$, the 
equations for the coefficients of $Q_{2j-1}$ are
\bea
\bsm x_{11}&x_{21}&\dots& x_{j1}\esm H^{NW}_j&=&\bsm 1&0&\dots&0\esm;\label{2j-11}\\
\bsm x_{12}&x_{22}&\dots& x_{j2}\esm H^{SE}_j&=&\bsm 1&0&\dots&0\esm,\label{2j-12}
\eea
where $x_{k1}$ and $x_{k2}$ are the non-zero elements in the first and second
rows of the coefficient of $\l^k$ in $Q_{2j-1}$.

As remarked in  \eqref{l-quotient} and \eqref{M-quotient}, we can reconstruct the beam data 
$\{l_j\}$, $\{M_j\}$ from the leading coefficients $\la Q_k\ra$ of the polynomials $\{Q_k\}$.
For $Q_{2j-1}$, we want to compute $x_{j1}$ and $x_{j2}$ in \eqref{2j-11}.  By Cramer's
rule, $x_{j1}$ can be obtained by replacing the last row of the matrix in \eqref{2j-11} by
the right-hand side of \eqref{2j-11} and computing the determinant. The same procedure for
\eqref{2j-12} gives
\be\label{2j-1top}
x_{j1}\= (-1)^{j-1}\frac{\D^{NW}_{j1}}{\D_j},\qquad x_{j2}\=
(-1)^{j-1}\frac{\D^{SE}_{j1}}{\D_j},
\ee
where $\D^{NW}_{j1}$ is the determinant of 
$H^{NW}_j$ with the first column and last row eliminated, and similarly for $\D^{SE}_{j1}$. 
By Proposition \ref{Pade-properties}, since  $Q_{2j-1}$ has degree $j$, the leading coefficient 
$\la Q_{2j-1}\ra$ is diagonal if $j$ is odd and off-diagonal if $j$ is even.  Thus we have

\begin{proposition} The leading coefficient of
$Q_{2j-1}$ is 
\be\label{leading2j-1a}
\la Q_{2j-1}\ra\=\bsm\frac{\D^{NW}_{j1}}{\D_j} &0\\0&\frac{\D^{SE}_{j1}}{\D_{j}}\esm
\ee
if $j$ is odd,
\be\label{leading2j-1b}
\la Q_{2j-1}\ra \=-\bsm 0&\frac{\D^{NW}_{j1}}{\D_j}\\
 \frac{\D^{SE}_{j1}}{\D_j}&0\ \esm
\ee 
if $j$ is even. \end{proposition}  

\mds We turn now to consideration of the coefficients of $Q_{2j}$.
In line with \eqref{Q2plus-reduction}, we introduce two $d\times d$ matrices
that are obtained by reorganizing $H$ after removing the first two columns and last
two rows:
\bes
H^{NE}\=\bsm c_1&a_2&c_3&a_4\ \dots\\
             a_2&b_3&a_4&b_5\ \dots\\
             c_3&a_4&c_5&a_6\ \dots\\
             a_4&b_5&a_6&b_7\ \dots\\
             &&\ddots&\esm;\quad\ \ 
H^{SW}\=\bsm b_1&a_2&b_3&a_4\ \dots\\
             a_2&c_3&a_4&c_5\ \dots\\
             b_3&a_4&b_5&a_6\ \dots\\
             a_4&c_5&a_6&c_7\ \dots\\ 
             &&\ddots&\esm.
\ees
Let $H^{NE}_j$ and $H^{SW}_j$ be the $j\times j$ principal minors of $H^{NE}$ and $H^{SW}$,
respectively.  The equations for the coefficients of $Q_{2j}$ are
\bea
\bsm x_{11}&x_{21}&\dots& x_{j1}\esm H^{NE}_j&=&-\bsm a_0&b_1&\dots&a_{j-1}\esm;
\label{2j1}\\
\bsm x_{12}&x_{22}&\dots& x_{j2}\esm H^{SW}_j&=&-\bsm a_0&c_1&\dots&a_{j-1}\esm.
\label{2j2}
\eea
Here $x_{k1}$ and $x_{k2}$ are the non-zero entries in the first and second rows of the
coefficient of $\l^k$ in $Q_{2j}$. Replacing the last row of $H^{NE}_J$ by the negative of
the right-hand side of \eqref{2j1}, then moving that to be the first row, gives
$H^{NW}_j$.  Applying the same reasoning to \eqref{2j2}, we obtain
\be\label{2j-top}
x_{j1}\=(-1)^j\frac{\D_j}{\D^{NE}_{j}} ;\qquad x_{j2}\=(-1)^j\frac{\D_j}{\D^{SW}_{j}}.
\ee
By Lemma \ref{Pade-properties}, since $Q_{2j}$ has degree $j$, the top coefficient is 
off-diagonal if $j$ is odd and diagonal if $j$ is even.  Therefore

\begin{proposition} The leading coefficient of $Q_{2j}$ is
\be\label{leading2ja}
\la Q_{2j}\ra\=-\bsm 0&\frac{\D_j}{\D^{NE}_j}\\ \frac{\D_j}{\D^{SW}_j}&0\esm
\ee
if $j$ is odd, 
\be\label{leading2jb}
\la Q_{2j}\ra\=\bsm \frac{\D_j}{\D^{NE}_j}&0\\ 0&\frac{\D_j}{\D^{SW}_j}\esm
\ee
if $j$ is even.\end{proposition}

\mds
We are now in a position to compute the data $\{l_k\}$, $\{M_k\}$, via
\eqref{l-quotient}, \eqref{M-quotient}. Note that 
\be\label{NW-SW}
H^{NW}_{j1}\=H^{SW}_{j-1};
\qquad H^{SE}_{j1}\=H^{NE}_{j-1}
\ee
We use these identities  to rewrite \eqref{leading2j-1a} and \eqref{leading2j-1b}.

If $j$ is odd, we have
\beas
l_{d-j+1}\bo&=&\la Q_{2j-1}\ra\la Q_{2j-2}\ra\inv\\
&=&\bsm\frac{\D^{SW}_{j-1}}{\D_j} &0\\
 0&\frac{\D^{NE}_{j-1}}{\D_j}\esm 
\bsm \frac{\D_{j-1}}{\D^{NE}_{j-1}}&0\\ 0&\frac{\D_{j-1}}{\D^{SW}_{j-1}}\esm\inv
\=\frac{\D^{SW}_{j-1}\D^{NE}_{j-1}}{\D_j\D_{j-1}}\,\bo.
\eeas
If $j$ is even, we have
\bes
l_{d-j+1}\bo\=\bsm 0&\frac{\D^{SW}_{j-1}}{\D_j}\\
 \frac{\D^{NE}_{j-1}}{\D_j}&0\ \esm 
\bsm 0&\frac{\D_{j-1}}{\D^{NE}_{j-1}}\\ \frac{\D_{j-1}}{\D^{SW}_{j-1}}&0\esm\inv
\=\frac{\D^{SW}_{j-1}\D^{NE}_{j-1}}{\D_j\D_{j-1}}\,\bo.
\ees

If $j$ is odd, we have
\beas
M_{d-j+1}&=&\la Q_{2j}\ra\la Q_{2j-1}\ra\inv\\
&=&-\bsm 0&\frac{\D_j}{\D^{NE}_j}\\ \frac{\D_j}{\D^{SW}_j}&0\esm
\bsm\frac{\D^{SW}_{j-1}}{\D_j} &0\\0&\frac{\D^{NE}_{j-1}}{\D_{j}}\esm\inv\=
-\bsm0& \frac{\D_j^2}{\D^{NE}_{j}\D^{NE}_{j-1}}\\ \frac{\D_j^2}{\D^{SW}_j\D^{SW}_{j-1}}&0\esm.
\eeas
If $j$ is even, we have
\bes
M_{d-j+1}\=-\bsm \frac{\D_j}{\D^{NE}_j}&0\\ 0&\frac{\D_j}{\D^{SW}_j}\esm
\bsm 0&\frac{\D^{SW}_{j-1}}{\D_j}\\\frac{\D^{NE}_{j-1}}{\D_j}&0\ \esm\inv\=
-\bsm0& \frac{\D_j^2}{\D^{NE}_{j}\D^{NE}_{j-1}}\\ \frac{\D_j^2}{\D^{SW}_j\D^{SW}_{j-1}}&0\esm.
\ees

\section{Appendix: The consistency conditions; the smooth case}\label{app1} 
\renewcommand{\theequation}{A.\arabic{equation}}

Under an additional assumption of smoothness, the  compatibility conditions that relate
\be\label{Lax1} 
D_x^2\Phi\ =\ (\bo+\l M)\Phi,\qquad M\ =\ \bsm 0&n\\ m&0\esm.
\ee
and
\be\label{Lax2} 
D_t\Phi\= [bD_x+a]\Phi,
\ee
namely
\bes
D_tD_x^2\Phi\=D_x^2D_t\Phi\
\ees
lead to
\bea
\l M_t\Phi&=&\left\{b_{xx}+2a_x+\l[b,M]\right\}D_x\Phi\nonumber\\
&&\ +\left\{a_{xx} + 2b_x +\l (bM)_x+\l b_xM+\l[a,M]\right\}\Phi.\label{compat}
\eea
At a given value of $t$, this is a differential equation for $\Phi$ of order at most one.  
We are assuming that $\Phi$ is  a solution of a nontrivial second--order equation.  
We assume that the differential operator in \eqref{compat}
trivializes: 
\bea
0&=&b_{xx}+2a_x+\l[b,M],\label{compat1}\\
\l M_t&=&a_{xx}+ 2b_x+\l (b M)_x+\l b_xM+ \l [a,M],\label{compat2}
\eea
since otherwise the system is degenerate.

As in Section \ref{2CH} we suppose that 
\bes
a\= a_0+\l\inv a_1,\qquad b\= b_0+\l\inv b_1,
\ees
and that $a_j$ and $b_j$ are bounded, $j=0,1$. Each equation in \eqref{compat1}, \eqref{compat2} 
leads to three equations, for the coefficients of the powers $\l^k$, $k=-1,0,1$.  

\medskip
For $k=-1$ the equations are
\be\label{minusone}
(a_1)_{xx}+2(b_1)_x\= 0\,; \qquad (b_1)_{xx}+2(a_1)_x\ = 0.
\ee
Thus 
\begin{equation} 
(a_1)_x=-\frac{(b_1)_{xx}}{2},   \qquad (b_1)_{xxx}-4 (b_1)_x=0.    
\end{equation}

The second equation implies that $b_1=C_1 e^{2x}+C_2 e^{-2x} +C_3$ and since 
$b_1$ is bounded, $b_1$ is a constant matrix, and by the first equation so is $a_1$.  

\sms
For $k=1$ the equations are
\be\label{plusone}
0\ =\ [b_0,M]\,;\qquad
M_t\ =\ (b_0M)_x+(b_0)_xM+[a_0,M].
\ee
We assume that $m\ne n$, so first equation in \eqref{plusone} implies that the diagonal part 
of $b_0$ is a multiple of the identity matrix, and the off-diagonal part is a multiple of $M$:
\be\label{b0}
b_0\ =\ uI+pM.
\ee
Therefore the diagonal part of the second equation in \eqref{plusone}
gives
\bea
0&=& (pmn)_x+(pn)_xm+(a_0)_{12}m-(a_0)_{21}n;\nonumber \\
0&=& (pmn)_x+(pm)_xn+(a_0)_{21}n-(a_0)_{12}m.\label{k=1}
\eea
Adding these two equations gives
\bes
0\=4p_x(mn)+3p(mn)_x.
\ees
Multiplying by $p^3(mn)^2$ gives $0=[p^4(mn)^3]_x$, so $p^4(mn)^3$ is constant.    
If $p\ne0$ then this is  a nontrivial
{\it a priori} relationship between $m$ and $n$.  Therefore we assume $p=0$.  With this
assumption, equations \eqref{k=1} imply that the off-diagonal part of $a_0$ is proportional
to $M$. We can write
\be\label{a''}
b_0\ =\ u\bo,\qquad a_0\ =\ \frac12\bsm w(x)+v(x)&0\\
0&w(x)-v(x)\esm + qM.
\ee
The remaining information from equations \eqref{compat1},
\eqref{compat2}
is contained in the equations for the $k=0$  part:
\bea
0&=& (b_0)_{xx}+2(a_0)_x+[b_1,M]\,;\label{zero1}\\[2mm]
0&=& (a_0)_{xx}+ 2 (b_0)_x+b_1M_x+[a_1,M].\label{zero2},
\eea

since $a_1$, $b_1$ are constant. Write
\bes
a_1\ =\ \bsm \a_1&\a_2\\ \a_3&\a_4\esm\,;
\qquad b_1\ =\ \bsm \b_1&\b_2\\ \b_3&\b_4\esm.
\ees
Looking at the diagonal terms, then the off-diagonal terms, in \eqref{zero1}, we find 
\bea
u_{xx}+w_x&=& 0;\qquad v_x\ = \ \b_3n-\b_2m;\label{diag1}\\
2(qn)_x+(\b_1-\b_4)n&=&\ 0\ =\ 2(qm)_x+(\b_4-\b_1)m\label{off1}.
\eea
Multiply the left side of \eqref{off1} by $n$, the right side by $m$, and add, to obtain:
\bes
0\=2\left[(qn)_xm+(qm)_xn\right]\=2\left[2q_x(mn)+q(mn)_x\right].
\ees

As above, unless $q=0$ this gives an {\it a priori\/} relation $q^2(mn)=$ constant.
Taking $q=0$, \eqref{off1} implies $\b_1=\b_4$.

Looking at the off-diagonal terms in \eqref{zero2}, we obtain equations for $n$ and for $m$ with
constant coefficients:
\bes
\b_1 n_x\= (\a_4-\a_1)n;\qquad \b_1 m_x\=(\a_1-\a_4)m.
\ees
In order to avoid trivial cases, we must assume that $\b_1=0$ and $\a_1=\a_4$.
Computing the diagonal part of \eqref{zero2}, taking into account
\eqref{diag1} gives
\bea
0&=&- \frac12 u_{xxx}+ 2 u_x+ (\b_2m+\b_3n)_x]\,;\label{diag3}\\[2mm]
0&=&\a_2m-\a_3n. \label{diag4}
\eea

To avoid a trivial linear relation 
between $m$ and  $n$ we need the off--diagonal terms $\a_2$, $\a_3$ of $a_1$ to vanish. 

\smallskip
Summing up to this point:
\beas
a&=& \frac12\bsm \g-u_x+v&0\\ 0&\g-u_x-v\esm 
      +\frac1{2\l}\bsm \a& 0\\ 0&\a\esm\,;\\
b&=& \bsm u&0\\ 0&u\esm+\frac1\l\bsm 0&\b_2\\ \b_3&0\esm\,,
\eeas
where $\g$,$\a$, $\b_2$, and $\b_3$ are constant. 

Keeping in mind the obvious
symmetry between $(\vp_1, m)$ on one hand and $(\vp_2,n)$ on  the other, we symmetrize 
by taking $\b_2=\b_3=\b$.  Moreover, the first Lax equation \eqref{Lax1} has an additional gauge 
symmetry $\Phi \rightarrow \omega(t,\lambda) \Phi$.   Under this gauge transformation 
$$ 
a\rightarrow \omega_t \omega^{-1} +a, 
$$
and thus, by choosing $\omega$  to satisfy $\omega_t+\frac12 (\gamma +\frac{\alpha}{\l})\omega=0$,  we can eliminate 
both $\alpha$ and $\gamma$ from the parametrization of $a$, obtaining:  
\beas
a&=& -\frac12\bsm u_x-v&0\\ 0&\-u_x+v\esm 
      ;\\
b&=& \bsm u&0\\ 0&u\esm+\frac1\l\bsm 0&\b\\ \b&0\esm.
\eeas
Finally, we note that this form of $a,b$ implies that the Lax pair has a scaling symmetry: 
$\lambda\rightarrow s\lambda, M\rightarrow \frac1s M, \beta\rightarrow s \beta$.   Choosing the scale to be 
$s=\frac1\beta$ fixes $\b=1$  and we obtain the final form (see \eqref{eq:abBCH})
\beas
a&=& -\frac12\bsm u_x-v&0\\ 0&\-u_x+v\esm 
      ;\\
b&=& \bsm u&0\\ 0&u\esm+\frac1\l\bsm 0&1\\ 1&0\esm, 
\eeas
used in Section \ref{2CH}.   


\section*{Acknowledgements}
We would like to thank Professors Xiangke Chang and Shihao Li for pointing out several misprints in the first draft of the paper.  

After the current work was completed we learned that \autoref{eq:evolBCH} was also derived in \cite{geng-wang}.  We would like to thank 
Professor Nianhua Li for bringing this reference to our attention.  

Jacek Szmigielski's research is supported by the Natural Sciences and Engineering Research Council of Canada (NSERC).


\def\cydot{\leavevmode\raise.4ex\hbox{.}}
  \def\cydot{\leavevmode\raise.4ex\hbox{.}}

\end{document}